\documentclass[aps,pra,twocolumn,groupedaddress,10pt]{revtex4-1}
\usepackage[english]{babel}
\usepackage{amsmath,amsthm,amssymb,graphicx,enumerate,bbm,xcolor,physics}
\usepackage{hyperref,soul}
\hypersetup{linkcolor=carmine,citecolor=darkgreen,urlcolor=googleblue,
anchorcolor=OliveGreen}

\usepackage{courier}
\usepackage{listings}
\newcommand{\one}{\mathbbm{1}}

\newtheorem{definition}{Definition}
\newtheorem{proposition}{Proposition}
\newcommand{\dop}[1]{{\mathsf{d}[#1]}}
\newcommand{\Acal}{\mathcal{A}}
\newcommand{\Bcal}{\mathcal{B}}
\newcommand{\Xcal}{\mathcal{X}}
\newcommand{\Ycal}{\mathcal{Y}}

\newcommand{\Bcalov}{{\overline{\mathcal{B}}}}

\newcommand{\Ycalov}{{\overline{\mathcal{Y}}}}
\newcommand{\Hcal}{\mathcal{H}}
\newcommand{\Hop}[1]{{\mathsf{H}(#1)}}
\newcommand{\Dop}[1]{{\mathsf{D}(#1)}}
\newcommand{\Hplusop}[1]{{\mathsf{H}_+(#1)}}
\newcommand{\Isf}{\mathsf{1}}
\newcommand{\Csf}{\mathsf{C}}
\newcommand{\Qsf}{\mathsf{Q}}
\newcommand{\Cbb}{\mathbbm{C}}

\newcommand{\Top}[1]{{\mathsf{T}[#1]}}

\newcommand{\Jbf}{{\bf J}}

\begin{document}

\title{Type-independent Characterization of Spacelike Separated Resources}

\author{Denis Rosset}
\affiliation{Perimeter Institute for Theoretical Physics, 31 Caroline Street North, Waterloo, Ontario Canada N2L 2Y5}

\author{David Schmid}
\affiliation{Perimeter Institute for Theoretical Physics, 31 Caroline Street North, Waterloo, Ontario Canada N2L 2Y5}
\affiliation{Institute for Quantum Computing and Department of Physics and Astronomy, University of Waterloo, Waterloo, Ontario N2L3G1, Canada}

\author{Francesco Buscemi}
\affiliation{Graduate School of Informatics, Nagoya University, Chikusa-ku, 464-8601 Nagoya, Japan}

\begin{abstract}
  Quantum theory describes multipartite objects of various types: quantum states, nonlocal boxes, steering assemblages, teleportages, distributed measurements, channels, and so on.
  Such objects describe, for example, the resources shared in quantum networks.
  Not all such objects are useful, however. In the context of space-like separated parties, devices which can be simulated using local operations and shared randomness are useless, and it is of paramount importance to be able to practically distinguish useful from useless quantum resources.
  Accordingly, a body of literature has arisen to provide tools for witnessing and quantifying the nonclassicality of objects of each specific type.
  In the present work, we provide a framework which subsumes and generalizes all of these resources, as well as the tools for witnessing and quantifying their nonclassicality.
\end{abstract}

\maketitle

Space-like separated resources of various types are studied in quantum information.
In the bipartite setting, at least ten types of objects have been considered as resources in computation or communication tasks: density matrices~{\cite{Horodecki2009}}, shared randomness~\cite{Fine1982}, nonlocal boxes~{\cite{Brunner2014}}, steering assemblages~{\cite{Cavalcanti2017}}, teleportages~{\cite{Cavalcanti2017a,Hoban2018}}, distributed POVMs (or `semiquantum' channels)~{\cite{Buscemi2012}}, device-independent steering channels~{\cite{Cavalcanti2013a}}, channel assemblages~{\cite{Piani2015}}, Bob-with-input steering assemblages~{\cite{Sainz2019}}, and bipartite quantum channels~\cite{watrous_2018}.
Heterogenous objects appear in quantum networks~{\cite{Cavalcanti2015}}; for example, depending on the local operations available,
various schemes of quantum key distribution have been proposed~\cite{wei2019}.
With the partial exception of Ref.~{\cite{Hoban2018}}, no unified framework has been given to describe and characterize\footnote{Ref.~\cite{Hoban2018} also recognized the value of treating various types of resources as channels, but did not study the full range of possibilities, as we have done here. The chief differences between our work and Ref.~6 are that our work focuses on nonclassicality (while Ref.~\cite{Hoban2018} focuses on postquantumness), and that we leverage the resource theoretic framework to unify the distinct types of nonclassicality, as discussed in the main text. }
 all these different multipartite space-like separated resource types.

Here, we work in the context of space-like separation, where local operations and shared randomness (LOSR operations) are free, but no-signaling forbids classical communication, and we provide a framework which unifies the study of nonclassicality of arbitrary types of resources. We introduce a common notation for all resource types, distinguished by the nature (trivial, classical, or quantum) of the input and output systems, and define a unified notion of nonclassicality which subsumes the natural notions of nonclassicality for every type of resource.

\begin{figure}[t]
  \includegraphics[scale=1]{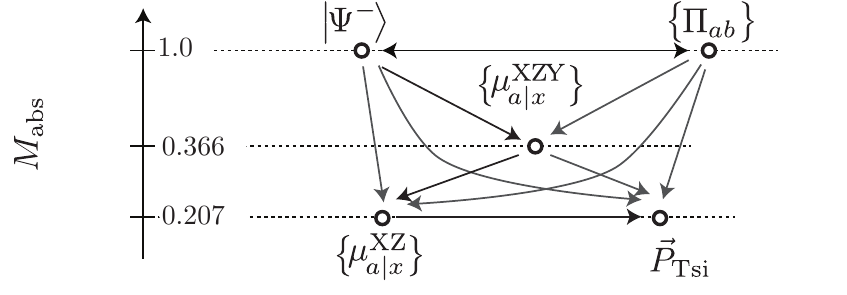}
  \begin{tabular}{llcc}
   Notation \ & Resource $R$  & \ Type $\mathsf{T}[R]$ & \ \ $M_{\text{abs}}(R)$  \\
    \hline
 $|\Psi^{-}\rangle$  & Singlet & $\mathsf{I} \mathsf{I} \rightarrow \mathsf{Q}
    \mathsf{Q}$ & 1\\
 $\left\{ \Pi_{ab} \right\}$  & BM$^2$ distributed POVM & $\mathsf{QQ} \rightarrow \mathsf{CC}$ & 1\\
 $\left\{ \mu^\text{XZY}_{a|x} \right\}$  & XZY-singlet assemblage & $\mathsf{CI} \rightarrow \mathsf{CQ}$ & $ \approx 0.366$ \\
 $\left\{ \mu^\text{XZ}_{a|x}\right\}$  & XZ-singlet assemblage & $\mathsf{CI} \rightarrow \mathsf{CQ}$ & $\approx 0.207$ \\
 $\vec{P}_{\rm Tsi}$  & Tsirelson box & $\mathsf{CC} \rightarrow \mathsf{CC}$ & $\approx 0.207$
  \end{tabular}
  \caption{\label{Fig:AbsRobPreorder}
    We exhaustively characterize the nonclassicality of five resources of four different types by their convertibility relations under LOSR operations (figure).
    We also characterize their nonclassicality using a {\em type-independent} absolute robustness monotone $M_\text{abs}$ (table) that we introduce below. Any single monotone is only partially informative; e.g., $M_\text{abs}$ assigns the same value to $\{ \mu^\text{XZ}_{a|x}\}$ and $\vec{P}_{\rm Tsi}$, even though the former is {\em strictly} more nonclassical than the latter.
}
\end{figure}

That is, we here (together with Ref.~\cite{Schmid2019}) demonstrate that entangled states, nonlocal boxes, unsteerable assemblages, nonclassical teleportages, and indeed resources of {\em all} the types just listed can be viewed as instances of a single notion of resourcefulness within a unified resource theory~\cite{Coecke2016}: that of {\em nonclassicality of common cause processes}.
This phrase was first introduced in Ref.~\cite{wolfe2019quantifying} in the context of a resource theory of nonlocal boxes, and it is also apt in the broader type-independent context considered here. This unified view based on LOSR operations also resolves some long-standing confusions~\cite{LOSRvsLOCCentang}.


While~\cite{Schmid2019} defines and studies this resource theory more abstractly, this paper proposes tools to quantify the nonclassicality of resources in a type-independent manner, and defines nonclassicality measures that transcend types.
For example, we will here quantitatively compare resources of various different types that have previously been studied only separately in the literature.
We also discuss how these measures can be computed or approximated using off-the-shelf software.
We base our discussion on the resources presented in Fig.~\ref{Fig:AbsRobPreorder}, whose features are elaborated on below.
In other words, we show that one can rigorously compare the quantitative degree of nonclassicality inherent in distributed resources of arbitrary types.

\paragraph*{A unified notation for all types of resources ---}
First, we introduce a notation which is capable of describing a wide variety of resources that arise naturally in Bell scenarios.
In our framework, a {\em resource} is a quantum device distributed over multiple space-like separated parties which receives inputs and produces outputs.
Viewed broadly as a channel, this subsumes a wide variety of special cases in the literature, as we will show.
For simplicity, we restrict ourselves to the two-party case and to finite-dimensional Hilbert spaces: we denote by $\Xcal$, $\Ycal$ the input systems and by $\Acal, \Bcal$ the output systems of the first (Alice) and second (Bob) parties respectively, as shown in Fig.~\ref{Fig:Resource}.
We also use $\Hcal$ as an auxiliary system or placeholder.

To each system, say $\Hcal$, we associate two values: the {\em dimension} $\dop{\Hcal}$ and the {\em type} $\Top{\Hcal}$ so that $\Hcal = (\dop{\Hcal}, \Top{\Hcal})$.
The dimension $\dop{\Hcal}$ describes the associated Hilbert space $\Cbb^{\dop{\Hcal}}$, equipped with the computational basis $\{\ket{i_\Hcal}\}_{i = 1}^{\dop{\Hcal}}$.
We write the set of Hermitian operators on $\Cbb^{\dop{\Hcal}}$ as $\Hop{\Hcal}$, the set of positive semidefinite operators as $\Hplusop{\Hcal} = \left\{ \rho \in \mathsf{H} (\mathcal{H}) : \rho \geqslant 0 \right\}$, and the set of density matrices as $\Dop{\Hcal} = \{\rho \in \mathsf{H}_+ (\mathcal{H}) : \tr(\rho) = 1 \}$.
The type $\Top{\Hcal} \in \{\Isf, \Csf, \Qsf\}$ describes whether the system is trivial, classical or quantum.
A system is {\em trivial} ($\Top{\Hcal} = \Isf$) if and only if $\dop{\Hcal} = 1$.
A {\em classical} system ($\Top{\Hcal} = \Csf$) restricts the operators in $\Hop{\Hcal}$ to a commutative subalgebra; without loss of generality, we take these operators to be diagonal in the computational basis $\ket{i_\Hcal}$.
Finally, a {\em quantum} system ($\Top{\Hcal} = \Qsf$) has no such restriction.
We sometimes omit the tensor product symbol: e.g., $\mathsf{D} (\mathcal{X}\mathcal{Y})$ denotes $\mathsf{D} (\mathcal{X} \otimes \mathcal{Y})$, and we omit subscripts labeling systems when convenient and unambiguous.

A resource $R$ is a map
\begin{equation}
  \label{Eq:TwoPartyResource}
  R_{\Acal\Bcal|\Xcal\Ycal} :
  \mathsf{D} (\mathcal{X} \mathcal{Y}) \rightarrow \mathsf{D} (\mathcal{A} \mathcal{B})\;,
\end{equation}
which is trace-preserving
\begin{equation}
  \label{Eq:TP}
 \tr(R [\xi \otimes \psi]) = 1 \quad \forall \xi \in \mathsf{D} (\mathcal{X}), \psi \in \mathsf{D} (\mathcal{Y})
\end{equation}
and completely positive~\cite{NielsenAndChuang,Schmidcausal}
\begin{equation}
  \label{Eq:CP}
(R\otimes \mathbb{I}) [\nu_{\mathcal{XYH'}}] \geqslant 0 \quad \forall \nu_{\mathcal{XYH'}} \in \mathsf{D} (\mathcal{XYH'})\;,
\end{equation}
for any auxilliary space $\Hcal'$.
The subscript $\Acal\Bcal|\Xcal\Ycal$ corresponds to the types and dimensions of all input and output systems.
For a resource $R$ as in~\eqref{Eq:TwoPartyResource}, we define the {\em resource type} written
\begin{equation}
  \label{Eq:ResourceType}
  \Top{R} \quad = \quad \mathsf{T} [\mathcal{X}] \mathsf{T} [\mathcal{Y}] \rightarrow \mathsf{T} [\mathcal{A}] \mathsf{T} [\mathcal{B}]\;,
\end{equation}
and the {\em resource dimension}
\begin{equation}
  \label{Eq:ResourceDimension}
  \dop{R} = (\dop{\Acal}, \dop{\Bcal}, \dop{\Xcal}, \dop{\Ycal})\;.
\end{equation}

We consider only resources which are {\em nonsignaling} from every party to every other.
No-signaling from Alice to Bob implies that if one ignores Alice's output $\mathcal{A}$, then Alice's input $\mathcal{X}$ has no influence on Bob's output $\mathcal{B}$.
That is, the reduced channel $R_{\mathcal{B}|\mathcal{X}\mathcal{Y}} \equiv \tr_{\mathcal{A}} R_{\mathcal{A}\mathcal{B}|\mathcal{X}\mathcal{Y}}$ of such a resource $R$ satisfies $R_{\mathcal{B}|\mathcal{X}\mathcal{Y}} [\xi \otimes  \psi] = R_{\mathcal{B}|\mathcal{X}\mathcal{Y}} [\xi' \otimes \psi]$ for all inputs $\xi, \xi' \in \mathsf{D} (\mathcal{X})$ and $\psi \in \mathsf{D} (\mathcal{Y})$.
Hence, one can pick an arbitrary $\xi$ to construct the unique $R_{\mathcal{B} | \mathcal{Y}}$ such that
\begin{equation}
  \label{Eq:RNonsignalingBXY} R_{\mathcal{B}|\mathcal{X}\mathcal{Y}} [\xi
  \otimes \psi] = \operatorname{tr} [\xi] R_{\mathcal{B} | \mathcal{Y}} [\psi]
\end{equation}
for all $\xi$ and $\psi$.
No-signaling from Bob to Alice can be encoded similarly as
\begin{equation}
  \label{Eq:RNonsignalingAXY} R_{\mathcal{A}|\mathcal{X}\mathcal{Y}} [\xi
  \otimes \psi] = \operatorname{tr} [\psi] R_{\mathcal{A} | \mathcal{X}} [\xi].
\end{equation}

\begin{figure}
  \includegraphics{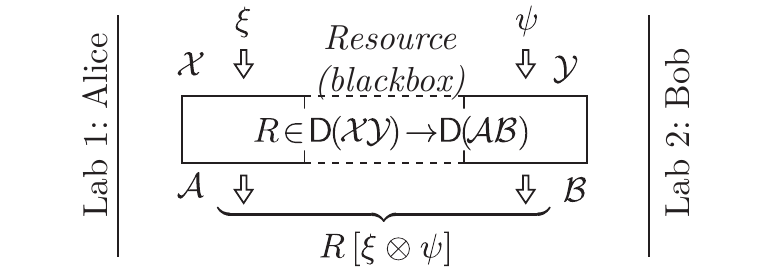}
  \caption{
    \label{Fig:Resource}
    A two-party resource $R \in \mathsf{D}(\mathcal{XY})\rightarrow\mathsf{D}(\mathcal{AB})$ and its action on product input $\xi \otimes \psi$; by allowing $\mathcal{X}, \mathcal{Y}, \mathcal{A}, \mathcal{B}$ to be trivial, classical, or quantum, we unify all the resource types shown in Table~\ref{Tab:ResourceKinds}.
    }
\end{figure}

\paragraph*{Input/output types ---}
Note that when an output $\mathcal{A}$ of a party is trivial, with $\mathsf{T} [\mathcal{A}] = \mathsf{I}$ (or equivalently $\mathsf{d} [\mathcal{A}] = \mathsf{I}$), no-signaling guarantees that the party is irrelevant and can always be ignored.
In contrast, when an input $\mathcal{X}$ of a party is trivial, with $\mathsf{T} [\mathcal{X}] = \mathsf{I}$, that party can nonetheless be nontrivial, as happens (e.g.) with quantum states.
When an output $\mathcal{A}$ is classical, with $\mathsf{T}[\mathcal{A}] = \mathsf{C}$, the channel satisfies (for all $\xi$ and all $\psi$):
  \begin{equation}
    \label{Eq:ClassicalOutput} \forall i_{\mathcal{A}} \neq j_{\mathcal{A}}, \qquad \langle i_{\mathcal{A}} | R [\xi \otimes \psi] | j_{\mathcal{A}} \rangle = 0.
  \end{equation}
  When an input $\mathcal{X}$ is classical, with $\mathsf{T}
  [\mathcal{X}] = \mathsf{C}$, the channel acts on the diagonal subspace, and
  thus satisfies (for all $\psi$)
  \begin{equation}
    \label{Eq:ClassicalInput} \forall i_{\mathcal{X}} \neq j_{\mathcal{X}}, \qquad R [| i_{\mathcal{X}} \rangle \langle j_{\mathcal{X}} | \otimes \psi] = 0.
  \end{equation}

\begin{table*}[t]
  \begin{tabular}{llclllc}
    Name & Type $\mathsf{T}[R]$ & Drawing &  & Name & Type $\mathsf{T}[R]$ & Drawing \\ \cline{1-3} \cline{5-7}
    Quantum state~\cite{Horodecki2009,LOSRvsLOCCentang} & $\mathsf{II} \rightarrow \mathsf{QQ}$ & \begin{minipage}{2cm} \includegraphics[page=1]{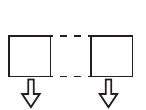} \end{minipage} &  &
    \parbox[c]{5cm}{\flushleft{Distributed measurement \\ (see semiquantum games~\cite{Buscemi2012})}} & $\mathsf{QQ} \rightarrow \mathsf{CC}$ & \begin{minipage}{2cm}\includegraphics[page=6]{boards.pdf}\end{minipage} \\
    Shared randomness~\cite{Fine1982} & $\mathsf{II} \rightarrow \mathsf{CC}$ & \begin{minipage}{2cm}\includegraphics[page=2]{boards.pdf}\end{minipage} &  &
    MDI steering assemblage~\cite{Cavalcanti2013a} & $\mathsf{CQ} \rightarrow \mathsf{CC}$ & \begin{minipage}{2cm}\includegraphics[page=7]{boards.pdf}\end{minipage} \\
    Nonlocal box~\cite{Brunner2014} & $\mathsf{CC} \rightarrow \mathsf{CC}$ & \begin{minipage}{2cm}\includegraphics[page=3]{boards.pdf}\end{minipage} &  &
    Channel assemblage~\cite{Piani2015} & $\mathsf{CQ} \rightarrow \mathsf{CQ}$ & \begin{minipage}{2cm}\includegraphics[page=8]{boards.pdf}\end{minipage} \\
    Steering assemblage~\cite{Cavalcanti2017} & $\mathsf{CI} \rightarrow \mathsf{CQ}$ & \begin{minipage}{2cm}\includegraphics[page=4]{boards.pdf}\end{minipage} &  &
    \parbox[c]{5cm}{\flushleft{Bob-with-input steering \\ assemblage~\cite{Sainz2019}}} & $\mathsf{CC} \rightarrow \mathsf{CQ}$ & \begin{minipage}{2cm}\includegraphics[page=9]{boards.pdf}\end{minipage} \\
    \parbox[c]{4cm}{\flushleft{Teleportage~\cite{Hoban2018} \\ (in teleportation exps.~\cite{Cavalcanti2017a})}} & $\mathsf{QI} \rightarrow \mathsf{CQ}$ & \begin{minipage}{2cm}\includegraphics[page=5]{boards.pdf}\end{minipage} &  &
    General bipartite channel~\cite{watrous_2018} & $\mathsf{QQ} \rightarrow \mathsf{QQ}$ & \begin{minipage}{2cm}\includegraphics[page=10]{boards.pdf}\end{minipage}
  \end{tabular}
  \caption[Types of resources studied in the literature]{\label{Tab:ResourceKinds}
    Types of resources studied in the literature.
    The arrow~\includegraphics{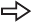} indicates a quantum input/output space, while the arrow~\includegraphics{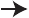} indicates a classical input/output space.
  }
\end{table*}

\begin{definition}
  \label{Def:NSResources}
The set ${\bf R}_{\mathcal{A}\mathcal{B}|\mathcal{X}\mathcal{Y}}$ of all {\em nonsignaling resources} of given types $\Top{R}$ and dimensions $\dop{R}$ is defined as those channels which satisfy
Eqs.~(\ref{Eq:TP})-(\ref{Eq:RNonsignalingAXY})
and (when applicable) (\ref{Eq:ClassicalOutput}),(\ref{Eq:ClassicalInput}).
\end{definition}
For any given type and dimensionalities, this set is representable via a semidefinite program (SDP), using the Choi matrix~\cite{Choi1975,Jamiolkowski1972} representation of resources; see Sec.~\ref{App:NSRepresentable} for an explicit proof.

We show ten distinct types of resources which our framework subsumes and which have been previously studied in the literature in Table~\ref{Tab:ResourceKinds}.

\paragraph*{Examples of resources ---}
We now define the examples from Fig.~\ref{Fig:AbsRobPreorder}.

The {\em singlet} $\ket{\Psi^-}$ is a quantum state (type $\mathsf{II} \rightarrow \mathsf{QQ}$), written as a channel acting on a trivial input, denoted $1$:
\begin{equation}
  R_{\ket{\Psi^-}}[1] \! = \! \dyad{\!\Psi^-\!} = \big(\ket{01}-\ket{10}\big)\big(\bra{01}-\bra{10}\big)/2\;.
\end{equation}

The {\em BM$^2$ distributed POVM} $\{ \Pi_{ab} \}$ (type $\mathsf{QQ} \rightarrow \mathsf{CC}$) is inspired by semiquantum tests of entanglement~\cite{Verbanis2016,Rosset2018a,Schmid2019}.
It can be constructed by two parties who share a singlet state $\ket{\Psi^-}$; Alice jointly measures system $\mathcal{X}$ together with her half of the singlet using a Bell measurement (BM) in the basis $\{ (\sigma_a \otimes \one) \ket{\Psi^-} \}_{a=0,1,2,3}$, where $a$ labels her outcome, and Bob proceeds similarly on system $\mathcal{Y}$ together with his half of the singlet, with $b=0,1,2,3$ as his outcome. (Here, $\sigma_0 = \one$ and $\sigma_{1,2,3}$ are the Pauli matrices.)
The action of the resulting distributed POVM with quantum inputs $\mathcal{X}$ and $\mathcal{Y}$ on input quantum states $\xi$ and $\psi$ can be written
\begin{equation}
  R_{\{ \Pi_{ab} \}}[\xi \otimes \psi] \! = \! \sum_{ab}\! \dyad{ab} \expval{\left( \sigma_a \xi \sigma_a \right) \otimes \left(\sigma_b \psi \sigma_b \right)}{\!\Psi^-\!}.
\end{equation}

The resources $\{ \mu^\text{XZY}_{a|x} \}$ and $\{ \mu^\text{XZ}_{a|x} \}$ are steering assemblages (type $\mathsf{CI} \rightarrow \mathsf{CQ}$).
The {\em XZY-singlet assemblage} $\left\{ \mu^\text{XZY}_{a|x} \right\}$ can be constructed by two parties sharing a singlet, and one of them performing the measurement
$\{ \dyad{+}, \dyad{-} \}$, $\{\ket{0}\bra{0}, \ket{1}\bra{1} \}$ or $\{ \ket{+\!i}\bra{+\!i}, \ket{-\!i}\bra{-\!i} \}$, depending on whether the classical input $x$ is $0$, $1$, or $2$, respectively, and where $\ket{\pm} = \frac{1}{\sqrt{2}}(\ket{0} \pm \ket{1}$) and $\ket{\pm\!i} = \frac{1}{\sqrt{2}}( \ket{0} \pm i \ket{1}$).
One can write the resulting assemblage in terms of its action on the three possible classical inputs as
\begin{align}
  R_{\{\mu^\text{XZY}_{a|x}\}}[\dyad{0}] & = \big(\dyad{0-} + \dyad{1+}\big)/2,  \\
  R_{\{\mu^\text{XZY}_{a|x}\}}[\dyad{1}] & = \big(\dyad{01} + \dyad{10}\big)/2, \nonumber  \\
  R_{\{\mu^\text{XZY}_{a|x}\}}[\dyad{2}] & = \big(\dyad{0-\!i} + \dyad{1+\!i}\big)/2. \nonumber
\end{align}
The {\em XZ-singlet assemblage} $\{ \mu^\text{XZ}_{a|x} \}$ is defined similarly, but where the classical inputs are $x=0,1$, corresponding to the first two POVMs above, respectively.

The {\em Tsirelson box} $P_\text{Tsi}(ab|xy)$ (type $\mathsf{CC} \rightarrow \mathsf{CC}$) of Ref~\cite{Tsirelson1987} is the quantumly-realizable box which maximally violates the CHSH inequality; it is obtained from $\ket{\Psi^-}$ by projective measurements.
Alice uses the same measurements as in the preparation of $R_{\{\mu^\text{XZ}_{a|x}\}}$, while Bob measures either $\{ c \ket{0} + s \ket{1}, s \ket{0} - c \ket{1} \}$ or $\{s \ket{0} + c \ket{1}, c \ket{0} - s \ket{1} \}$, depending on whether $y=0$ or $y=1$, respectively, and where $c = \cos \pi/8$ and $s = \sin \pi/8$.
Then with $a,b,x,y = 0,1$, we have $R_\text{Tsi}[\dyad{xy}] = \sum_{ab} \dyad{ab} P_\text{Psi}(ab|xy)$ with
\begin{equation} \label{Eq:PRdefn}
  P_\text{Tsi}(ab|xy) = \frac{1 + (-1)^{a + b + xy}/\sqrt{2}}{4}.
\end{equation}

\paragraph*{A unified resource theory ---}
Next, we introduce a single {\em resource theory}~\cite{Coecke2016}  which captures the relevant notion of nonclassicality for {\em all} resources of all of the types described above.
Within this resource theory (which is expanded upon in the companion article~\cite{Schmid2019}), free resources are the ones that are obtained using only {\em local operations and shared randomness} (LOSR operations~\cite{wolfe2019quantifying,LOSRvsLOCCentang}):
\begin{definition}
  \label{Def:FreeResources}
  A resource $R_{\mathcal{AB}|\mathcal{XY}}$ is {\em LOSR-free} if it admits of a convex decomposition into single party resources $R_{\mathcal{A}|\mathcal{X}}^i$ and
  $R_{\mathcal{B}|\mathcal{Y}}^i$ according to probability distribution $p_i$:
  \begin{equation}
    \label{Eq:RConvexDecomposition}
    R_{\mathcal{A}\mathcal{B}|\mathcal{X}\mathcal{Y}} = \sum_i p_i (R_{\mathcal{A}|\mathcal{X}}^i \otimes R_{\mathcal{B}|\mathcal{Y}}^i)\; .
  \end{equation}
\end{definition}

We denote the set of {\em all} free resources as $\mathbf{R}^{\rm free}:= \bigcup_{\mathcal{A}\mathcal{B}\mathcal{X}\mathcal{Y}} {\bf R}^{\rm free}_{\mathcal{A}\mathcal{B} | \mathcal{X}\mathcal{Y}}$, and the set of {\em all} no-signaling resources (free or nonfree) as $\mathbf{R} := \bigcup_{\mathcal{A}\mathcal{B}\mathcal{X}\mathcal{Y}} {\bf R}_{\mathcal{A}\mathcal{B} | \mathcal{X}\mathcal{Y}}$; note that the union runs over all types and dimensions.

For some types, we have ${\bf R}^{\rm free}_{\mathcal{A}\mathcal{B} | \mathcal{X}\mathcal{Y}} = {\bf R}_{\mathcal{A}\mathcal{B} | \mathcal{X}\mathcal{Y}}$.

\begin{definition}
  \label{Def:Useless}
  A resource type is {\em $\mathsf{T}$-trivial} if every resource of that type is necessarily free.
\end{definition}

\begin{proposition}
  \label{Prop:Useless}
  Any type $\mathsf{T} [\mathcal{X}] \mathsf{T} [\mathcal{Y}] \rightarrow \mathsf{T} [\mathcal{A}] \mathsf{T} [\mathcal{B}]$ with $\mathsf{T}[X] = \mathsf{I}$ and $\mathsf{T}[A] = \mathsf{C}$ is $\mathsf{T}$-trivial.
\end{proposition}
\begin{proof}
  Let $R$ have type $\mathsf{I} \mathsf{T}[\mathcal{Y}] \to \mathsf{C} \mathsf{T}[\mathcal{B}]$.
  Let $P(a) = \expval{\tr_{\mathcal{B}} R[\psi]}{a}$ for an arbitrary $\psi \in \mathsf{D}(\mathcal{Y})$. By no-signaling, $P(a)$ is independent of $\psi$ and unique. For $a$ with $P(a) > 0$, define the single-party channel $R_{a}[\_] = \expval{R[\_]}{a}/P(a)\in \mathsf{D}(\mathcal{Y})\! \to \!\mathsf{D}(\mathcal{B})$.
Noting that $R[\psi] = \sum_a P(a) \dyad{a} \!\otimes\! R_a[\psi]$, it follows that $R$ is LOSR-free.
\end{proof}

We now turn our attention towards transformations of resources: a generic transformation $\tau$ on resources is a completely positive, linear supermap~\cite{Chiribella2008}
\begin{equation}
  \label{Eq:Transformation}
  \tau: \left[ \mathsf{D} (\mathcal{X} \mathcal{Y}) \rightarrow \mathsf{D} (\mathcal{A} \mathcal{B}) \right ] \rightarrow
  \left [ ( \mathsf{D} (\mathcal{X}' \mathcal{Y}') \rightarrow \mathsf{D} (\mathcal{A}' \mathcal{B}') \right ]
  \;
\end{equation}
that transforms a resource $R_{\mathcal{A}\mathcal{B}|\mathcal{X}\mathcal{Y}}$ into a resource $R_{\mathcal{A}'\mathcal{B}'|\mathcal{X}'\mathcal{Y}'}$, possibly changing its type and dimensions.
A transformation is {\em LOSR-free} if it is obtainable by local operations and shared randomness, and hence admits of a convex decomposition into products of arbitrary supermaps~\cite{Chiribella2008} acting only on a single party. A generic (bipartite) free transformation is shown in Fig.~3(a) of Ref.~\cite{Schmid2019}.
The set of free transformations is closed under composition, and maps free resources to free resources.

The resourcefulness of resources is completely characterized by the conversions that are possible between them using free operations.
For example, the relative nonclassicality of the resources in Fig.~\ref{Fig:AbsRobPreorder} are fully characterized by the conversion relations claimed therein, whose validity we now prove.
When we defined the examples above, we described how the singlet $\ket{\Psi^-}$ can be transformed into the resources $\{ \Pi_{ab} \}$, $\{\mu^\text{XZY}_{a|x}\}$, $\{\mu^\text{XZ}_{a|x}\}$, and $\vec{P}_{\rm Tsi}$ using local (and hence LOSR-free) transformations.
Furthermore, $\{ \Pi_{ab} \}$ can be freely transformed into the singlet $\ket{\Psi^-}$, as proved in~\cite{Schmid2019}, so the two are equally resourceful (and hence can be freely transformed into exactly the same resources).
The assemblage $\{\mu^\text{XZY}_{a|x}\}$ can be freely transformed into assemblage $\{\mu^\text{XZ}_{a|x}\}$ simply by Alice ignoring the third input value, and then also into $P_\text{Tsi}(ab|xy)$ by Bob further performing the measurements given above Eq.~\eqref{Eq:PRdefn} on his quantum output.
Finally, $\{\mu^\text{XZ}_{a|x}\}$ can be transformed into $P_\text{Tsi}(ab|xy)$ by Bob performing these same measurements.

It remains to show that every arrow absent from the figure corresponds to a conversion that is impossible under LOSR-free transformations.
The following proposition implies the nonconvertibility of $P_\text{Tsi}$ into $\{\mu^\text{XZY}_{a|x}\}$, $\{\mu^\text{XZ}_{a|x}\}$, or $\ket{\Psi^-}$, as well as the nonconvertibility of either $\{\mu^\text{XZY}_{a|x}\}$ or $\{\mu^\text{XZ}_{a|x}\}$ into $\ket{\Psi^-}$. By transitivity, these imply the nonconvertibility of any of $P_\text{Tsi}$, $\{\mu^\text{XZY}_{a|x}\}$, and $\{\mu^\text{XZ}_{a|x}\}$ into $\{ \Pi_{ab} \}$.
\begin{figure}[t]
  \includegraphics{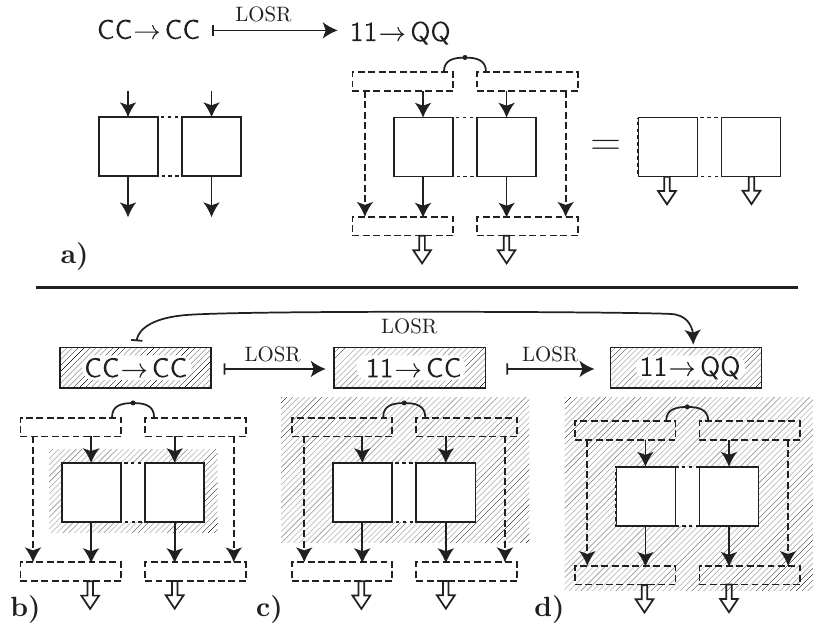}
  \caption{\label{proof}
Graphical proof of case 1. of Prop.~\ref{Prop:NonclassicalityDegradingConversions}.
    (a) The most general transformation from type $ \mathsf{C}\mathsf{C} \rightarrow \mathsf{C}\mathsf{C}$ to type $ \mathsf{I}\mathsf{I} \rightarrow  \mathsf{Q}\mathsf{Q}$; without loss of generality, the side channels can be taken to be classical.
One can imagine the transformation in (a) as a two-step procedure: the initial resource (b) of type $\mathsf{C}\mathsf{C} \rightarrow \mathsf{C}\mathsf{C}$ is transformed to a resource (c) of type $ \mathsf{I}\mathsf{I} \rightarrow  \mathsf{C}\mathsf{C}$, and then to a resource (d) of type $ \mathsf{I}\mathsf{I} \rightarrow  \mathsf{Q}\mathsf{Q}$.
In (b)-(d), the shading indicates the portion of the figure whose type is labeled.}
\end{figure}
\begin{proposition}
  \label{Prop:NonclassicalityDegradingConversions}
 For any transformation taking
  \begin{enumerate}
  \item boxes $(\mathsf{CC}\to\mathsf{CC})$ to states $(\mathsf{II}\to\mathsf{QQ})$,
  \item boxes $(\mathsf{CC}\to\mathsf{CC})$ to assemblages $(\mathsf{CI}\to\mathsf{CQ})$, or
  \item assemblages $(\mathsf{CI}\to\mathsf{CQ})$ to states $(\mathsf{II}\to\mathsf{QQ})$,
  \end{enumerate}
  the resulting resource is necessarily LOSR-free.
\end{proposition}
\begin{proof}
  Consider first a generic LOSR transformation taking boxes $(\mathsf{CC}\to\mathsf{CC})$ to states $(\mathsf{II}\to\mathsf{QQ})$; the most general such transformation is depicted by the dashed operations in Fig.~\ref{proof}a), since for any party with only a classical input and output, the side channels on their local supermaps may be taken to be classical systems without loss of generality. As shown in Fig.~\ref{proof}b)-Fig.~\ref{proof}d), any such LOSR transformation can be seen as a composition of two LOSR transformations, where the first takes type $\mathsf{C}\mathsf{C} \rightarrow \mathsf{C}\mathsf{C}$ to type $\mathsf{I}\mathsf{I} \rightarrow \mathsf{C}\mathsf{C}$, and the second takes type $\mathsf{I}\mathsf{I} \rightarrow \mathsf{C}\mathsf{C}$ to type $\mathsf{I}\mathsf{I} \rightarrow \mathsf{Q}\mathsf{Q}$.
Now, Prop.~\ref{Prop:Useless}
guarantees that the result of the first transformation is free, since its type is $\mathsf{T}$-trivial.
  Since LOSR operations preserve the free set, the final resource resulting after the second transformation is also free.

  The other two cases listed in the proposition have analogous proofs, but where the $\mathsf{T}$-trivial intermediate types are $\mathsf{CI}\to\mathsf{CC}$ and $\mathsf{II}\to\mathsf{CQ}$, respectively.
\end{proof}

To establish the nonconvertibility of $\{ \mu^\text{XZ}_{a|x} \}$ to $\{ \mu^\text{XZY}_{a|x} \}$, it suffices to exhibit a monotone whose value would increase under the conversion.
Such a monotone is the type-independent absolute robustness, which we introduce below and depict in Fig.~\ref{Fig:AbsRobPreorder}.

\paragraph*{Type-independent monotones ---} One can quantitatively measure the nonclassicality of a resource using any function that does not increase under LOSR transformations, which is a {\em monotone} of our resource theory.
We focus on the type-independent absolute robustness here, but consider others in Sec.~\ref{App:Monotones}.

\begin{definition}
  \label{Def:AbsoluteRobustness} The {\em type-independent absolute robustness} $M_{\rm abs} (R)$ of a resource $R \in {\bf R}$ of arbitrary type and dimension is:
  \begin{multline}
    M_{\rm abs} (R) = \min s \ \text{such that } \ (R + s S)/(1 + s) \in {\bf R}^{\rm free},\nonumber\\
                            s \geqslant 0, \quad S \in {\bf R}^{\rm free}, \quad \dop{S} = \dop{R}, \Top{S} = \Top{R}.
  \end{multline}
\end{definition}

Our innovation here is to consider functions which behave monotonically {\em even under operations that change the resource type}, which allows us to compare the nonclassicality of resources across different types, as in Figure~\ref{Fig:AbsRobPreorder}.
We prove in Sec.~\ref{App:Monotones} that the absolute robustness (as defined here) has this property--- despite the fact that the special cases of it which have been previously studied~\cite{Vidal1999,Geller2014,Cavalcanti2016a,Sainz2016,Vicente2014,wolfe2019quantifying} have been type specific, and despite the fact that the computation of $M_{\text{abs}} (R)$ involves a specification of the type and dimension of $R$.

The values of this monotone on the examples in Fig.~\ref{Fig:AbsRobPreorder} are exact, and the manner by which they are computed is described in Sec.~\ref{App:AbsoluteRobustnessComputation}.  (We also compute these values for a family of parameterized versions of four of these resources.)
Note that monotones can be used to prove that some conversions under LOSR are impossible, namely, those which would increase the value of $M_\text{abs}$. One {\em cannot} conclude anything about which conversions {\em are} possible. This can be seen in Fig.~\ref{Fig:AbsRobPreorder}, where $M_\text{abs}$ assigns the same value to the Tsirelson box and to the assemblage $\{ \mu^\text{XZ}_{a|x}\}$, despite the fact that the two are {\em not} equivalent, as the former cannot be freely converted to the latter by Prop.~\ref{Prop:NonclassicalityDegradingConversions}.

%


\paragraph*{A hierarchy to characterize nonclassicality ---}
We now describe how one can {\em in practice} determine whether or not a given resource $R^{\star}$ is LOSR-free.
This can be done using a hierarchy of SDPs which ultimately checks if $R^{\star}$ is a member of ${\bf R}^{\text{free}}$. The reasoning is as follows.

If $R^{\star}_{\Acal \Bcal | \Xcal \Ycal} \in {\bf R}^{\text{free}}$, then $R^{\star}$ has a convex decomposition as in~(\ref{Eq:RConvexDecomposition}).
By copying the shared randomness to more parties (who can then locally emulate any other party) it follows that $R^{\star}$ has an $n$-symmetric extension~\cite{Doherty2004,Berta2018} for any $n$, obtained by copying the second party $n$ times in the product.
That is, for any $n$ there exists
\begin{equation}
  \label{Eq:ExplicitSymmetricExtension}
  R^{(n)} = \sum_i p_i (R_{\mathcal{A}|\mathcal{X}}^i \otimes R_{\mathcal{B}_1 |\mathcal{Y}_1}^i \otimes \ldots \otimes R_{\mathcal{B}_n |\mathcal{Y}_n}^i)
\end{equation}
such that

\begin{enumerate}[1.]
  \item $R^{(n)}$ is no-signaling from any party to any other,
  \item $R^{(n)}$ is symmetric under all $n!$ permutations of the copies, and
  \item the reduced resource $\tr_{\mathcal{B}_2 \ldots \mathcal{B}_k} R^{(n)} = R^\star$.
\end{enumerate}
\noindent Whether or not an $n$-symmetric extension $R^{(n)}$ exists can be tested by an SDP. Hence:

\begin{proposition}
  \label{Prop:SymmetricExtensions}
  The set of classical resources ${\bf R}^{\text{free}}_{\mathcal{A}\mathcal{B} |  \mathcal{X}\mathcal{Y}}$ of any given type and dimensionalities has a sequence of outer approximations
  \begin{equation}
    \label{Eq:ResourceHierarchy}
    {\bf R}_{\mathcal{A}\mathcal{B} | \mathcal{X}\mathcal{Y}} \supseteq {\bf F}_{\mathcal{A}\mathcal{B} | \mathcal{X}\mathcal{Y}}^{(1)} \supseteq \ldots {\bf F}_{\mathcal{A}\mathcal{B} | \mathcal{X}\mathcal{Y}}^{(n)} \ldots \supseteq {\bf R}^{\text{free}}_{\mathcal{A}\mathcal{B}|\mathcal{X}\mathcal{Y}}
  \end{equation}
  where each ${\bf F}_{\mathcal{A}\mathcal{B} | \mathcal{X}\mathcal{Y}}^{(n)}$ is representable by a SDP, such that $\lim_{n \rightarrow \infty} {\bf F}_{\mathcal{A}\mathcal{B} | \mathcal{X}\mathcal{Y}}^{(n)} = {\bf R}_{\mathcal{A}\mathcal{B} |  \mathcal{X}\mathcal{Y}}^{\text{free}}$.
\end{proposition}

\begin{proof}
  Each ${\bf F}_{\mathcal{A}\mathcal{B} | \mathcal{X}\mathcal{Y}}^{(n)}$ is defined by the set of resources that admit of an $n$-symmetric extension, and each such set is SDP-representable (see Def.~\ref{Def:NSResources}).
 Leveraging the Choi matrix representation of resources, convergence follows immediately from Ref.~{\cite[Theorem 3.4]{Berta2018}}, since the constraints of Eqs.~(\ref{Eq:TP})-(\ref{Eq:RNonsignalingAXY}) and (when applicable) (\ref{Eq:ClassicalOutput}),(\ref{Eq:ClassicalInput}) have the form required by the theorem; see Sec.~\ref{App:FreeRepresentable} for explicit details.
\end{proof}

\paragraph*{Application: witnessing nonclassicality ---}
Because this sequence of approximations converges on ${\bf R}_{\mathcal{A}\mathcal{B} | \mathcal{X}\mathcal{Y}}^{\text{free}}$, for any nonclassical resource $R \notin {\bf R}_{\mathcal{A}\mathcal{B} | \mathcal{X}\mathcal{Y}}^{\text{free}}$ of arbitrary type and dimension, there exists some level $n$ of the hierarchy such that $R \notin {\bf F}^{(n)}_{\mathcal{A}\mathcal{B} | \mathcal{X}\mathcal{Y}}$, which witnesses the fact that $R$ is not free.
We stress again that our technique applies to {\em all types of resources}, since all that distinguishes different types is the inclusion (or not) of constraints of the form in (\ref{Eq:ClassicalOutput})-(\ref{Eq:ClassicalInput}) in the relevant SDPs.

\paragraph*{Application: computation of monotones ---}
This hierarchy enables the practical computation of monotones.
For example, consider again the absolute robustness $M_{\text{abs}}$.
By replacing ${\bf R}_{\mathcal{A}\mathcal{B} | \mathcal{X}\mathcal{Y}}^{\text{free}}$ with the outer approximation ${\bf F}^{(n)}_{\mathcal{A}\mathcal{B} | \mathcal{X}\mathcal{Y}}$ in the definition of $M_{\text{abs}}$, one relaxes the constraints in the minimization, thus obtaining a {\em lower bound} on $M_\text{abs} (R)$.
This gives a sequence of approximations
\begin{equation}
  \tilde{M}_{\text{abs}}^{(1)}(R) \leqslant \ldots \leqslant \tilde{M}_{\text{abs}}^{(n)}(R)  \leqslant \ldots \leqslant M_{\text{abs}}(R)
\end{equation}
each of which is explicitly computable using a SDP (see Sec.~\ref{App:MonotoneRepresentable}), such that $\lim_{n \rightarrow \infty} \tilde{M}_{\text{abs}}^{(n)}(R) = M_{\text{abs}} (R)$.

\paragraph*{Application: LOSR convertibility of states---}
Since convertibility relations fundamentally determine the resourcefulness of resources, it is of central importance to be able to determine when one resource can be freely converted to another.
By definition, free transformations that convert states (type $\mathsf{II}\to\mathsf{QQ}$) into states are themselves free resources of type $\mathsf{QQ}\to\mathsf{QQ}$.
Given states $\rho\in\mathsf{D}(\mathcal{AB})$ and $\rho'\in\mathsf{D}(\mathcal{A'B'})$, one can test for the existence of an LOSR-free transformation (viewed as a resource $R$) such that $  \rho' = R[\rho] $.
Since this is a linear constraint, one can modify the hierarchy of Proposition~\ref{Prop:SymmetricExtensions} to include it, generating a new hierarchy which tests for the possibility of LOSR convertibility between $\rho$ and $\rho'$.
A natural extension of this hierarchy to allow generic local supermaps~\cite{Chiribella2008} would allow one to test for the possibility of LOSR conversions between arbitrary types of resources, but we leave this for future work.

\paragraph*{Outlook ---}

We have given a type independent description of nonclassical resources when local operations and shared randomness are free, generalizing a series of existing results, including those in Ref.~\cite{Horodecki2009,Bowles2015,Brunner2014,Cavalcanti2017,Cavalcanti2017a,Hoban2018,Buscemi2012,Cavalcanti2013a,Piani2015,Sainz2019}.
This work opens many new research avenues.
First, although generalizing our definitions and arguments to $n$-party scenarios is straightforward, multipartite nonclassicality in general scenarios is likely to have a rich structure, as happens for multipartite LOCC-entanglement~\cite{Bancal2011,Gallego2012,Bancal2013}.
Still more interesting would be to expand our framework to encompass supermaps; for example, this would allow for a description of filtering~\cite{Dur2000}, and for the construction of a hierarchy of SDPs that can witness the possibility or impossibility of any LOSR conversion between any two resources of any type.
Much work remains to be done in defining monotones which are computable and which have operational significance across multiple resource types. These questions can also be asked for the resource theory of postquantumness, wherein the free operations are local operations and shared randomness~\cite{postquantumschmid}. For all these, the methods of Ref.~\cite{Gonda2019}, which presents a unified framework to study monotones in resource theories, are likely to be useful.
Finally, we hope that our framework will be used to unify the expanding set of scenarios under study, and to generalize the tools developed to characterize them.

\section*{Acknowledgments}

We acknowledge useful discussions with S.-L. Chen, R.W. Spekkens, and E. Wolfe.
D.S. is supported by a Vanier Canada Graduate Scholarship.
F.B. is grateful for the hospitality of Perimeter Institute where part of this work was carried out and acknowledges partial support from the Japan Society for the Promotion of Science (JSPS) KAKENHI, Grant No.19H04066 and No.20K03746.
Research at Perimeter Institute is supported in part by the Government of Canada through the Department of Innovation, Science and Economic Development Canada and by the Province of Ontario through the Ministry of Economic Development, Job Creation and Trade.
This publication was made possible through the support of a grant from the John Templeton Foundation. The opinions expressed in this publication are those of the authors and do not necessarily reflect the views of the John Templeton Foundation.

\section*{Appendices}

\appendix

We order the next sections by increasing complexity, and do not follow necessarily the order of appearance of notions in the main text.

In Section~\ref{App:Monotones}, we prove that the type-independent absolute robustness $M_\text{abs}$ (Definition~\ref{Def:AbsoluteRobustness} of the main text) is a monotone of our type-independent resource theory of local operations and shared randomness.
We also explore other monotones (generalized robustness, nonlocal weight), and link their type-independent variants to their type-specific appearances in the literature.
We also show that random robustness cannot naturally be generalized to be type-independent.
In Section~\ref{App:ConvexOptimization} we review basic notions of convex optimization, and define LP- and SDP-representable sets, as these are used several times in the present work.
Then, in Section~\ref{App:NSRepresentable}, we show that the set of nonsignaling resources of any given type and dimensionalities is SDP-representable (sometimes even LP-representable).
On the other hand, the representability of the set of free resources is more complex: we show in Section~\ref{App:FreeRepresentable} that some types can be represented exactly using LPs or SDPs, while other types are represented by a convergent hierarchy of SDPs.
In Section~\ref{App:MonotoneRepresentable}, we show how to compute or approximate monotones using these LP- and SDP representations, and we describe in Section~\ref{App:AbsoluteRobustnessComputation} how we computed the absolute robustness of the examples presented in the main text, while extending the examples to partially entangled states.

\section{Monotones}
\label{App:Monotones}

When comparing resources of a fixed type, it suffices to use the traditional notion of a monotone as a function (from resources of a fixed type to the real numbers) which does not increase under {\em type-preserving} free operations.
\begin{definition}
  A {\em type-dependent LOSR monotone} $\mathcal{M}: \bf{R} \rightarrow \mathbbm{R}$ is a function that
  obeys
  \begin{equation}
    \mathcal{M} (\tau[R]) \leqslant \mathcal{M} (R)
  \end{equation}
  for all resources $R$ and LOSR-free transformations $\tau$ such that $R$ and $\tau[R]$ have fixed, constant type $\Top{R} = \Top{\tau[R]}$.
\end{definition}

For example, the negativity~\cite{Vidal2002} is a type-dependent LOCC (and thus LOSR) monotone that applies to quantum states ($\Top{R} = \Isf\Isf\to\Qsf\Qsf$); the relative entropy of nonlocality~\cite{Dam2005,Vicente2014} applies to nonlocal boxes ($\Top{R} = \Csf\Csf\to\Csf\Csf$) and so on.
Note that some of these monotones impose also that the dimension of the resource stays constant ($\dop{R} = \dop{\tau[R]}$), as in the example of the random robustness measure, discussed below.

However, the main message of this article is that one can and should compare the resourcefulness of resources across arbitrary types and dimensions.
Hence, one should consider dimension- and type-independent monotones, defined formally as follows.

\begin{definition}
  A {\em type-independent LOSR monotone} is a function $  M :
  \mathbf{R} \rightarrow \mathbbm{R}$ from the set $\mathbf{R}$ of all resources (of any type) to the real numbers, such that
  \begin{equation}
    M(\tau [R]) \leqslant M(R)
  \end{equation}
  for any free LOSR transformation $\tau$, including those which change the type and dimension of an input resource.
\end{definition}

We defined and focused on one such monotone, the type-independent absolute robustness, in the main text.
We now define several more type-independent monotones, as natural generalizations of type-dependent monotones that have been previously studied~\cite{Vidal1999, Geller2014, Cavalcanti2016a, Sainz2016, Vicente2014, wolfe2019quantifying, Steiner2003, Geller2014, Piani2015a, Cavalcanti2016a, Gallego2015}.

\subsection{Type-independent absolute robustness monotone}
\label{App:AbsRobApprox}

We have already defined the type-independent {\em absolute robustness} $M_{\text{abs}}$ in the main text in Definition~\ref{Def:AbsoluteRobustness}, as
\begin{multline}
  \label{Eq:AbsoluteRobustness}
  M_{\rm abs} (R) = \min s \ \text{such that } \ (R + s S)/(1 + s) \in {\bf R}^{\rm free},\\
  s \geqslant 0, \quad S \in {\bf R}^{\rm free}, \quad \dop{S} = \dop{R}, \Top{S} = \Top{R}\;.
\end{multline}

We now prove that this is indeed a type-independent monotone.

\begin{proof}
  For a given resource $R_{\Acal\Bcal|\Xcal\Ycal}$, the value $M_{\text{abs}} (R)$ is the value $s$ of a feasible solution ($s$, $S$) satisfying
  \begin{equation}
    \frac{ R+ s  S}{1+s} \in {\bf R}^{\rm free}_{\Acal\Bcal|\Xcal\Ycal}\;,
  \end{equation}
  for some $S \in {\bf R}^{\rm free}_{\Acal\Bcal|\Xcal\Ycal}$.
  Since any free LOSR transformation
  $\tau : \bf{R}_{\Acal\Bcal|\Xcal\Ycal} \rightarrow \bf{R}_{\Acal'\Bcal' | \Xcal'\Ycal'}$ (which may change the resource type and dimension) takes the free set ${\bf R}^{\rm free}_{\Acal\Bcal| \Xcal\Ycal}$ into the free set ${\bf R}^{\rm free}_{\Acal'\Bcal'| \Xcal'\Ycal'}$, it follows by linearity that
  \begin{equation}
    \label{intermedstep}
    \tau \left[ \frac{ R+ s  S}{1+s} \right ] = \frac{\tau [R] + s \tau [S]}{1+s} \in {\bf R}^{\rm free}_{\Acal'\Bcal'| \Xcal'\Ycal'} \;,
  \end{equation}
  and also that $\tau[S] \in {\bf R}^{\rm free}_{\Acal'\Bcal'| \Xcal'\Ycal'}$.
  Hence $(s,\tau[S])$ is a feasible solution in the optimization~\eqref{Eq:AbsoluteRobustness} for $\tau[R]$, and so the value $s$ is an upper bound on $M_{\text{abs}} (\tau [R])$, and so $M_{\text{abs}} (\tau [R]) \le M_{\text{abs}} (R)$. Since $\tau$ may be arbitrarily type-changing, $M_{\text{abs}}$ is a type-independent monotone.   \end{proof}

This is the natural type-independent generalization of the absolute robustness monotone defined on quantum states in Ref.~\cite{Vidal1999}, defined on boxes in Ref.~\cite{Geller2014}, and defined on steering assemblages in Refs.~\cite{Cavalcanti2016a,Sainz2016}.
(It was termed the LHS steering robustness in~\cite{Cavalcanti2016a} and the LHS robustness in~\cite{Sainz2016}.)

A closely related type-independent monotone is obtained by using a different function to quantify the noise being mixed with the given resource:
\begin{multline}
  \label{Eq:VariantAbsoluteRobustness}
  M_{\rm abs'} (R) = \min s \ \text{such that } \ (1-s) R + s S \in {\bf R}^{\rm free},\\
  s \geqslant 0, \quad S \in {\bf R}^{\rm free}, \quad \dop{S} = \dop{R}, \Top{S} = \Top{R}\;.
\end{multline}
 The value of this monotone is related by the variable transformation $M_{\text{abs}'}(R) = M_{\text{abs}}(R) / (1 + M_{\text{abs}}(R))$.

 This variant is the natural type-independent generalization of the monotone $M_{\text{abs}'}$ defined on boxes in Refs.~\cite{Vicente2014,wolfe2019quantifying}. (It was termed the robustness of nonlocality in Ref.~\cite{Vicente2014} and termed $M_{\text{RBST},\bf{L}}$ in Ref.~\cite{wolfe2019quantifying}.)

\subsection{Type-independent generalized robustness monotone}

Another related type-independent monotone is obtained by considering mixing a different set of resources, namely, by replacing $S \in {\bf R}^{\rm free}$ with $S \in {\bf R}$.
Then, in analogy with the generalized robustness introduced for quantum states in~\cite{Steiner2003}, we define:
\begin{multline}
  \label{Eq:GeneralizedRobustness}
  M_{\rm gen} (R) = \min s \ \text{such that } \ (R + s S)/(1 + s) \in {\bf R}^{\rm free},\\
  s \geqslant 0, \quad S \in {\bf R}, \quad \dop{S} = \dop{R}, \Top{S} = \Top{R}\;.
\end{multline}
The proof that this is a type-independent monotone is analogous to the proof for the absolute robustness.
The definition above subsumes the definition for quantum states~\cite{Steiner2003}, boxes~\cite{Geller2014}, and steering assemblages~\cite{Gallego2015,Piani2015a,Cavalcanti2016a}. (Note that the same quantity was called the {\em steering robustness} in Refs.~\cite{Piani2015a,Cavalcanti2016a} and the {\em robustness of steering} in Ref.~\cite{Gallego2015}.)

Because the set ${\bf R}_{\Acal\Bcal|\Xcal\Ycal}$ is SDP-representable in general (unlike the set $ {\bf R}^{\rm free}_{\Acal\Bcal|\Xcal\Ycal}$), computing this monotone can be more efficient than computing the absolute robustness.

  Noting that the enveloping set ${\bf R}_{\Acal\Bcal|\Xcal\Ycal}$ of resources includes post-quantum resources, one might be motivated to impose a further constraint on the enveloping set, choosing it instead to be the set ${\bf R}_{\Acal\Bcal|\Xcal\Ycal}$ of quantumly realizable LOSR resources (sometimes termed localizable resources~\cite{Beckman2001,DAriano2011}), as was done in~\cite{Geller2014}; doing so then defines a distinct monotone.
  However, the characterization of the set of quantumly realizable LOSR resources is an open problem (and for the case of boxes, the set is even not closed~\cite{Slofstra2017a}).
  We leave as an open question the impact of the choice of a type-independent enveloping theory on the operational significance of the generalized robustness~\cite{Takagi2019a}.

\subsection{Type-independent nonlocal weight monotone}
Another related type-independent monotone is given by the minimum weight on a nonfree resource with which one can mix a free resource to generate the given resource:
\begin{multline}
  \label{Eq:NonlocalWeight}
  M_{\rm NLW} (R) = \min e \ \text{such that } \ R = e \tilde{R} + (1-e) N\;,\\
  e \geqslant 0, \quad \tilde{R} \in {\bf R}, \quad N \in {\bf R}^{\rm free}\;, \\
  \dop{R} = \dop{\tilde{R}} = \dop{N}, \quad   \Top{R} = \Top{\tilde{R}} = \Top{N}\;.
\end{multline}
The proof that it is a type-independent monotone is again analogous to the proof for the absolute robustness.

This is the natural type-independent generalization of the nonlocal weight monotone~\cite{Portmann2012} defined on boxes, and also of the steering weight~\cite{Cavalcanti2016a} or {\em steerable weight}~\cite{Skrzypczyk2014,Gallego2015} monotone, defined on steering assemblages.
It is also linked~\cite{Geller2014} to the {\em best separable approximation} of quantum states~\cite{Lewenstein1998}.

\subsection{Type- and dimension- {\em dependent} random robustness measure}

We close by noting that the random robustness monotone defined in Ref.~\cite{Vidal1999} for quantum states {\em cannot} be easily generalized to a type-independent monotone.
The random robustness is defined to parallel the absolute robustness, but where the set of resources being mixed with the given resource is taken to be a singleton set containing only the ``maximally mixed'' resource of the same type as the given resource.
However, the resulting measure is {\em not} a type-independent LOSR monotone; indeed, it is not even a type-{\em dependent} monotone under LOSR operations (nor under LO operations or LOCC operations, for that matter).
As shown in Ref.~\cite{Vidal1999}, the random robustness can increase under local operations that change the dimensionality of one's resources. Ultimately, this is a consequence of the dimension-dependence of the maximally mixed resource; e.g., for two-qubit states the random robustness is based on the maximally mixed state $\mathbbm{1}_4 / 4$, while for two-qutrit states it is based on $\mathbbm{1}_9 / 9$.

However, if one restricts attention only to resources of a fixed type and a fixed dimension, then one can show that the random robustness behaves monotonically.
It is conceivable that the random robustness satisfies some new notion of type-independence which incorporates dimensional constraints, but it is unclear whether such a new notion is of any significance, and we do not pursue it here.

\section{Convex optimization and LP- and SDP- representability}
\label{App:ConvexOptimization}

We now present some practical tools which are useful for characterizing various important sets of resources, leveraging the fact that efficient numerical solvers such as MOSEK~\cite{ApS2015} are available to represent and optimize over convex sets that are defined using linear inequalities (such solvers are termed linear programs, or LPs) and linear matrix inequalities (such solvers are termed semidefinite programs, or SDPs)~\cite{Helton2009}.

Both LPs and SDPs optimize a linear objective over a convex set: linear programs optimize over {\em LP-representable sets}, while semidefinite programs optimize over {\em SDP-representable sets}.

\begin{definition}
  \label{Def:LPRep}
  A set $\textbf{X} \subseteq \mathbbm{R}^n$ is {\em LP-representable} if
  \begin{equation}
    \textbf{X} = \left \{ \vec{x} \in \mathbbm{R}^n \middle | \exists \vec{y} \in \mathbbm{R}^m : \vec{\gamma} = \vec{c} + P \vec{x} + Q \vec{y} \geqslant 0 \right \},
  \end{equation}
  where $\vec{c} \in \mathbbm{R}^p$, $P \in \mathbbm{R}^{p \times n}$ and $Q \in \mathbbm{R}^{p \times m}$, and $\vec{\gamma} \geqslant 0$ is interpreted componentwise.
\end{definition}

Thus, LP-representable sets are defined by linear inequality constraints, linear equality constraints (which can be written as pairs of inequalities), and projections of such sets (by moving the corresponding coordinates to the slack variables $\vec{y}$); we refer the reader to~\cite{Boyd2004} for a complete reference on such manipulations.

\begin{definition}
  \label{Def:SDPRep}
  A set $\textbf{X} \subseteq \mathbbm{R}^n$ is {\em SDP-representable} if
  \begin{equation}
    \label{Eq:SDPRep}
    \textbf{X} \! = \! \left \{ \vec{x} \in \mathbbm{R}^n \middle | \exists \vec{y} \in \mathbbm{R}^m : \Gamma\! = \! C \! + \! \sum_{i=1}^n \! x_i P_i \! + \! \sum_{j=1}^m y_i Q_i \geqslant 0 \right \},
  \end{equation}
  where $C, \{ P_i \}_i, \{ Q_j \}_j$ are symmetric matrices of size $p \times p$, and, for a symmetric matrix $M \in \mathbbm{R}^{p \times p}$, the constraint $\Gamma \geqslant 0$ is interpreted as $\Gamma$ being semidefinite positive.
\end{definition}
Thus, SDP-representable sets are defined by equality constraints of various types, semidefinite positiveness constraints, projections, and all the operations that define LP-representable sets (the inequality constraint $\vec{\gamma} \geqslant 0$ is rewritten as the SDP constraint of a diagonal matrix).

A variety of sets are LP- or SDP-representable.
For example, although the set $\textbf{X}$ in the definition above contains real vectors, sets of Hermitian matrices can be represented as well, by expanding on a basis with real coefficients. Again, we refer the reader to~\cite{Boyd2004} for more details.

\begin{definition}
  A {\em linear program} (resp. {\em semidefinite program}) corresponds to the problem
  \begin{equation}
    \max_{\vec{x} \in \textbf{X}} \vec{b}^{\top} \vec{x}\;,
  \end{equation}
  where $\textbf{X}$ is LP-representable (resp. SDP-representable).
\end{definition}
It is also useful to consider conic extensions of convex sets.

\begin{definition}
  \label{Def:ConicExtension}
  Let $\textbf{X}$ be a convex set.
  The {\em conic extension} $\hat{\textbf{X}} \supseteq \textbf{X}$ is
  \begin{equation}
    \hat{\textbf{X}} = \left \{ \alpha \vec{x} | \alpha \in \mathbbm{R}, \alpha \ge 0, \vec{x} \in \textbf{X} \right \}.
  \end{equation}
\end{definition}

We easily verify that if $\textbf{X}$ is LP-representable, then its conic extension $\hat{\textbf{X}}$ is also LP-representable: it is sufficient to add an additional variable $t \in \mathbbm{R}$ with the constraint $t \ge 0$ to Definition~\eqref{Def:LPRep}, and replace $\vec{c} + P \vec{x} + Q \vec{y} \geqslant 0$ by $\vec{c} t + P \vec{x} + Q \vec{y} \geqslant 0$.
By a similar argument, one can verify that if $\textbf{X}$ is SDP-representable, then $\hat{\textbf{X}}$ is also SDP-representable.

\section{The set ${\bf R}_{\mathcal{A}\mathcal{B}|\mathcal{X}\mathcal{Y}}$ of all  nonsignaling resources of given types and dimensionalities is SDP-representable}
\label{App:NSRepresentable}

Next, we introduce the useful Choi description of resources, and then show that the set of all nonsignaling resources of any given type and dimensionalities can be represented exactly by an SDP or LP.

\subsection{Choi state representation of resources}

In the main text, we described resources as completely positive, trace preserving linear maps.
Here, we introduce a second (equivalent) description of a resource $R \in \mathsf{D} (\Xcal\Ycal) \rightarrow \mathsf{D} (\mathcal{A}\mathcal{B})$ in terms of its Choi~{\cite{Choi1975,Jamiolkowski1972}} matrix $J_R \in \mathsf{D} \left( \Acal\Bcal\Xcal'\Ycal' \right)$.
The Hilbert spaces $\Xcal'$ and $\Ycal'$ are isomorphic to $\Xcal$ and $\Ycal$ respectively.
The distinction between $\Xcal$ and $\Xcal'$ is introduced for clarity in this section, and will be dropped later on.
Using the computational basis of those spaces, we construct the maximally entangled states $\phi_{\Xcal \Xcal'}$ and $\phi_{\Ycal \Ycal'}$:
\begin{equation}
  \phi_{\Xcal \Xcal'} = \frac{1}{\dop{\Xcal}} \sum_{i, i'=1}^{\dop{\Xcal}} \dyad{ i_{\Xcal} i_{\Xcal'} }{ i'_{\Xcal} i'_{\Xcal'} },
\end{equation}
\begin{equation}
  \phi_{\Ycal \Ycal'} = \frac{1}{\dop{\Ycal}} \sum_{j, j'=1}^{\dop{\Ycal}} \dyad{ j_{\Ycal} j_{\Ycal'} }{ j'_{\Ycal} j'_{\Ycal'} } .
\end{equation}
\begin{definition}
  The {\em Choi state} $J_R \in \mathsf{D} \left( \Acal\Bcal\Xcal'\Ycal' \right)$ corresponding to the resource $R$ is defined as:
  \begin{equation}
    \label{Eq:Choi}
    J_R = \left( R_{\Acal\Bcal|\Xcal\Ycal} \otimes \one_{\Xcal'\Ycal'} \right) \left[ \phi_{\Xcal\Ycal'} \otimes \phi_{\Ycal\Ycal'} \right] .
  \end{equation}
\end{definition}
One can write this out explicitly as
\begin{multline}
  J_R = \frac{1}{\dop{\Xcal} \dop{\Ycal}} \sum_{i, i' = 1}^{\dop{\Xcal}} \sum_{j, j' = 1}^{\dop{\Ycal}}
  R \left [ \dyad{i_{\Xcal}}{i'_{\Xcal}} \otimes \dyad{ j_{\Ycal} }{ j'_{\Ycal} } \right ] \\
  \otimes
  \dyad{i_{\Xcal'} }{ i'_{\Xcal'} } \otimes \dyad{j_{\Ycal'} }{ j'_{\Ycal'} } \; .
\end{multline}

\subsection{Constraints on the Choi state of a resource with special properties}
The constraint that a channel be completely positive is equivalent to the constraint on the Choi state that
\begin{equation}
  \label{Eq:ChoiSDP}
  J_R \in \mathsf{H}_+ \left( \Acal\Bcal\Xcal'\Ycal' \right),
\end{equation}
while the constraint that a channel be trace-preserving is equivalent to the constraint on the Choi state that
\begin{equation}
  \label{Eq:ChoiTP}
  \tr_{\Acal\Bcal} J_R = \frac{1}{\dop{\Xcal'} \dop{\Ycal'}} \one_{\Xcal'\Ycal'} \;.
\end{equation}

As discussed in the main text, we consider only no-signaling resources, as befits our assumption of space-like separation.
This translates into the following constraints on the Choi state:
\begin{align}
  \label{Eq:ChoiNS}
  \tr_\Acal J_R & = \frac{ \one_{\Xcal'} \otimes \tr_{\Acal \Xcal'}( J_R ) }{\dop{\Xcal'}} \nonumber\\
  \tr_\Bcal J_R & = \frac{ \one_{\Ycal'} \otimes \tr_{\Bcal \Ycal'}( J_R ) }{\dop{\Ycal'}}
\end{align}

Note that condition~{\eqref{Eq:ChoiNS}} implies condition~{\eqref{Eq:ChoiTP}} above.
These nonsignaling constraints can be written explicitly as follows.
For all $j_\Bcal, j_{\Xcal'}, j_{\Ycal'}, k_{\Bcal}, k_{\Ycal'}$:
\begin{multline}
  \label{Eq:CNonsignaling2}
  \sum_{j_{\Acal}} \matrixel{
    j_{\Acal} j_{\Bcal} j_{\Xcal'} j_{\Ycal'}
  }{J_R}{
    j_{\Acal} k_{\Bcal} k_{\Xcal'} k_{\Ycal'}
  } =\\
\frac{\delta_{j_{\Xcal'}, k_{\Xcal'}}}{\dop{\Xcal'}}
\sum_{j_{\Acal} k_{\Xcal'}}
\matrixel{
  j_{\Acal} j_{\Bcal} k_{\Xcal'} j_{\Ycal'}
}{J_R}{
  j_{\Acal} k_{\Bcal} k_{\Xcal'} k_{\Ycal'}
  },
\end{multline}
and for all $j_A, j_{\Xcal'}, j_{\Ycal'}, k_{\Acal}, k_{\Xcal'}$:
\begin{multline}
  \sum_{j_{\Bcal}} \matrixel{ j_{\Acal} j_{\Bcal} j_{\Xcal'} j_{\Ycal'} }{J_R}{ k_{\Acal} j_{\Bcal} k_{\Xcal'} k_{\Ycal'} } \\
  = \frac{\delta_{j_{\Ycal'}, k_{\Ycal'}}}{\dop{\Ycal'}}\sum_{j_{\Bcal} k_{\Ycal'}}
  \matrixel{ j_{\Acal} j_{\Bcal} j_{\Xcal'} k_{\Ycal'} }{J_R}{ k_{\Acal} j_{\Bcal} k_{\Xcal'} k_{\Ycal'}
   } .
 \end{multline}

Resources which have certain inputs or outputs that are classical satisfy further constraints, as we now show.
Trivial input or output Hilbert spaces correspond to one-dimensional Hilbert spaces in the tensor product in which the Choi state $J_R$ is defined, and thus such systems do not appear explicitly in the description of $J_R$ as a Hermitian matrix, and so imply no constraints beyond fixing the dimensionality.

\subsubsection*{Constraints from classicality of outputs and inputs }
As discussed in the main text, the statement that an output $\Acal$ is classical, denoted $\mathsf{T} [\Acal] = \mathsf{C}$, expresses the fact that the given resource always produces outputs on $\Acal$ that are diagonal in a fixed basis, as in Eq.~\eqref{Eq:ClassicalOutput}.
The Choi state of such a resource satisfies $\matrixel{i_{\Acal}}{J_R}{j_{\Acal}} = 0$ for $i_{\Acal} \neq j_{\Acal}$. Explicitly, for all $i_{\Acal}, i_{\Bcal}, i_{\Xcal'}, i_{\Ycal'}, j_{\Acal}, j_{\Bcal}, j_{\Ycal'}$ and $j_{\Ycal'}$ we have:
\begin{equation}
  \label{Eq:ClassicalOutput1}
  i_{\Acal} \neq j_{\Acal}
  \quad \Rightarrow \quad
  \matrixel{
    i_{\Acal} i_{\Bcal} i_{\Xcal'} i_{\Ycal'}
    }{J_R}{
      j_{\Acal} j_{\Bcal} j_{\Xcal'} j_{\Ycal'}
    } = 0 \;.
\end{equation}
The constraints for the systems of other parties are analogous.

The statement that an input $\Xcal$ is classical, denoted $\mathsf{T} [\Xcal] = \mathsf{C}$, expresses the fact that the given resource is defined only for input states that are diagonal in a fixed basis, as in Eq.~\eqref{Eq:ClassicalOutput}.
The Choi state of such a resource satisfies $\matrixel{ i_{\Xcal'} }{ J }{ j_{\Xcal'} } = 0$ for $i_{\Xcal'} \neq j_{\Xcal'}$.
Explicitly, for all $i_{\Acal}, i_{\Bcal}, i_{\Xcal'}, i_{\Ycal'}, j_{\Acal}, j_{\Bcal}, j_{\Xcal'}$ and $j_{\Ycal'}$ we have:
\begin{equation}
  \label{Eq:ClassicalInput1}
  i_{\Xcal'} \neq j_{\Xcal'}
  \quad \Rightarrow \quad
  \matrixel{
    i_{\Acal} i_{\Bcal} i_{\Xcal'} i_{\Ycal'}
  }{J_R}{
    j_{\Acal} j_{\Bcal} j_{\Xcal'} j_{\Ycal'}
  } = 0.
\end{equation}

Hence we have the following.
\begin{definition}
  \label{Def:ChoiNSResources}
The set $\textbf{J}_{\Acal\Bcal | \Xcal\Ycal}$ of Choi states corresponding to nonsignaling resources $\textbf{R}_{\Acal\Bcal | \Xcal\Ycal}$ of a given type and dimensionalities is equal to the set of Hermitian matrices that are
  \begin{itemize}
  \item completely positive~\eqref{Eq:ChoiSDP},
  \item trace preserving~\eqref{Eq:ChoiTP},
  \item nonsignaling~\eqref{Eq:ChoiNS},
  \item classical in the relevant Hilbert spaces: \eqref{Eq:ClassicalOutput1} and~\eqref{Eq:ClassicalInput1}.
  \end{itemize}
\end{definition}

This proposition follows.
\begin{proposition}
  \label{Prop:NSSetRepresentable}
  The set $\textbf{R}_{\Acal\Bcal | \Xcal\Ycal}$ of nonsignaling resources of any given type and dimensionalities is SDP-representable, as is the set $\textbf{J}_{\Acal\Bcal | \Xcal\Ycal}$ of corresponding Choi states.
  In the case where no input or output has quantum type, these sets are furthermore {\rm LP}-representable.
\end{proposition}
\begin{proof}
All the constraints in Definition~\ref{Def:ChoiNSResources} are of the form that define SDP-representable sets.
Thus, the set $\textbf{J}_{\Acal\Bcal | \Xcal\Ycal}$ (for any given types and dimensionalities) is SDP-representable, as is the corresponding set $\textbf{R}_{\Acal\Bcal | \Xcal\Ycal}$.
  When no inputs or outputs have quantum type, the Choi matrix is fully diagonal, and hence the semidefinite constraint~\eqref{Eq:ChoiSDP} reduces to a linear inequality constraint, which means that all the constraints are LP-representable.
\end{proof}

\section{Representability of the set of LOSR-free resources of any type}
\label{App:FreeRepresentable}
In Section~\ref{App:NSRepresentable} we discussed the characterization of the set $\textbf{R}_{\Acal\Bcal | \Xcal\Ycal}$ of nonsignaling resources, by characterizing the set $\textbf{J}_{\Acal\Bcal | \Xcal\Ycal}$ of Choi states of nonsignaling resources.
We now turn our attention to the harder but more important task of characterizing the set of {\em LOSR-free} resources $\textbf{R}^{\rm free}_{\Acal\Bcal | \Xcal\Ycal}$, by characterizing the set $\textbf{J}_{\Acal\Bcal | \Xcal\Ycal}^{\text{free}}$ of the corresponding Choi states.

We first describe particular types for which the set of free resources can be described exactly using an LP or SDP.
We then describe in detail how, for any given type and dimensionalities, this set can be represented by an SDP hierarchy.

\subsection{Exactly LP- or SDP-representable resource types}
\label{App:ExactFreeRepresentable}

It is widely known that nonlocal boxes of type $\mathsf{CC} \rightarrow \mathsf{CC}$ have an exact representation using linear programs~{\cite{Fine1982,Garg1984,Pitowsky1991}}.
It is also known that steering assemblages of the form $\mathsf{CI}
\rightarrow \mathsf{CQ}$ can be represented using semidefinite programs~{\cite{Cavalcanti2017}}.
We generalize these results.
To simplify the discussion, we will in this section refer to a party with input $\Xcal$ and output $\Acal$ as a {\em classical party} $\Acal|\Xcal$ whenever $(\mathsf{T} [\Acal], \mathsf{T} [\Xcal]) \in \big \{ (\mathsf{C}, \mathsf{C}), (\mathsf{I}, \mathsf{C}), (\mathsf{C}, \mathsf{I}), (\mathsf{I}, \mathsf{I}) \big \}$.

We first give some special types for which the set of LOSR-free resources of that type admit of an exact LP or an exact SDP representation.

For single party resources, the situation is simple.

\begin{proposition}
  The set $\Jbf_{\Acal | \Xcal}^{\text{free}}$ of free single-party resources of any type is always SDP-representable.
Moreover, $\Jbf_{\Acal | \Xcal}^{\text{free}}$ is LP-representable when the party $\Acal|\Xcal$ is classical, that is, whenever neither the input nor the output of the party contains any quantum component.
\end{proposition}

\begin{proof}
  This follows directly from the fact that $\Jbf_{\Acal|\Xcal}^{\text{free}} = \Jbf_{\Acal | \Xcal}$; that is, the set of free single-party resources coincides with the set of all no-signaling single-party resources, which by Proposition~\ref{Prop:NSSetRepresentable} is SDP-representable.
  The LP-representable case follows from that proposition as well.
\end{proof}

One can then show the following; note that here, with respect to the notation introduced above, we dropped the prime superscripts on the input spaces.

\begin{proposition}
  \label{Prop:SDPRepresentable}
  Consider the $(N+1)$-party resource with systems $\Acal$, $\Bcal_1, \ldots, \Bcal_N$, $\Xcal$, $\mathcal{Y}_1, \ldots, \mathcal{Y}_N$, and define $\Bcalov = \Bcal_1 \otimes \ldots \otimes \Bcal_N$, and similarly for $\Ycalov$.
  Let the party $\Acal|\Xcal$ be classical and the other systems be of arbitrary type and dimension, with $\Jbf_{\Bcalov|\Ycalov}^{\text{free}}$ SDP-representable.
  Then $\Jbf_{\Acal\Bcalov|\Xcal\Ycalov}^{\text{free}}$ is also SDP-representable.
\end{proposition}

\begin{proof}
Any free $J_{\Acal\Bcalov\Xcal\Ycalov}$ can be written as
  \begin{equation} \label{Eq:decomp}
    J_{\Acal\Bcalov\Xcal\Ycalov} = \sum p_i J^i_{\Acal\Xcal} \otimes J^i_{\Bcalov\Ycalov}\;.
  \end{equation}
   When party $\Acal|\Xcal$ is classical, the set $\Jbf_{\Acal | \Xcal} = \Jbf_{\Acal | \Xcal}^{\text{free}}$ has a finite number of extremal points, one for each deterministic strategy. Denoting the finite set of associated Choi states, indexed by $\lambda$, as $\{ J^{\lambda}_{\Acal\Xcal} \}_{\lambda}$, one can rewrite Eq.~\eqref{Eq:decomp} (without loss of generality) as
     \begin{equation}
      J_{\Acal\Bcalov\Xcal\Ycalov} = \sum_{\lambda} q_{\lambda} J^{\lambda}_{\Acal\Xcal} \otimes J^{\lambda}_{\Bcalov\Ycalov}\;.
  \end{equation}
  for some probability distribution $q_{\lambda}$ and where (for all $\lambda$) $\left \{ \hat{J}^{\lambda}_{\Bcalov\Ycalov} \right \}_\lambda \subset \hat{\Jbf}_{\Bcalov|\Ycalov}^{\text{free}}$.
   By further defining  $\hat{J}^{\lambda}_{\Bcalov \Ycalov} := q_{\lambda} J^{\lambda}_{\Bcalov | \Ycalov}$, one can simplify this to
    \begin{equation} \label{Eq:linsum}
      J_{\Acal\Bcalov\Xcal\Ycalov} = \sum J^{\lambda}_{\Acal\Xcal} \otimes \hat{J}^{\lambda}_{\Bcalov\Ycalov}.
  \end{equation}
  Now, if $\Jbf_{\Bcalov|\Ycalov}^{\text{free}}$ is SDP-representable, then so too is the set of objects of the form $J_{\Acal\Bcalov\Xcal\Ycalov}$, since such objects are simply linear combinations of objects drawn from an SDP-representable set. (In particular, for each classical value that system $\Xcal$ and $\Acal$ can take, Eq.~\eqref{Eq:linsum} constitutes a linear combination of objects $\hat{J}^{\lambda}_{\Bcalov\Ycalov}$ which are each drawn from an SDP-representable set.)
\end{proof}

A similar proposition holds for LP-representability, with an analogous proof.
\begin{proposition}
  \label{Prop:LPRepresentable}
  Let party $\Acal|\Xcal$ be classical, and $\Jbf_{\Bcalov|\Ycalov}^{\text{free}}$ be LP-representable.
  Then $\textbf{J}_{\Acal\Bcalov|\Xcal\Ycalov}^{\text{free}}$ is also LP-representable.
\end{proposition}

One final special case worth noting is the type $\Isf\Isf \rightarrow \Qsf\Qsf$, which is SDP-representable when $\dop{\Acal} \dop{\Bcal} \leqslant 6$, by virtue of the PPT criterion~{\cite{Horodecki1996}}.

A few applications of these propositions are given in Table~\ref{Tab:ExactFreeSets}.

\begin{table}[t]
  \begin{tabular}{lll}
   Resource name  & Type & Repr. \\
    \hline
    Boxes & $\Csf\Csf \to \Csf\Csf$ & LP \\
    Assemblages & $\Csf\Isf \to \Csf\Qsf$ & SDP \\
    MDI steering assemblages & $\Csf\Qsf \to \Csf\Csf$ & SDP \\
    Bob w/input steering assemblages & $\Csf\Csf \to \Csf\Qsf$ & SDP \\
    Channel assemblages & $\Csf\Qsf \to \Csf\Qsf$ & SDP
  \end{tabular}
  \caption{
    \label{Tab:ExactFreeSets}
    Types with exact LP or SDP representability of their free sets.
  }
\end{table}

\subsection{A hierarchy of SDPs for the characterization of free resources of arbitrary types}
\label{App:Hierarchy}

The set of LOSR-free resources of more general types do not admit of exact LP or exact SDP representations, and so one must resort to a hierarchy of approximations, which we will now give, after introducing some useful preliminaries.

\subsubsection*{Symmetric extensions}

Let $J \in \mathsf{D} (\Acal\Xcal\Bcal\Ycal)$ be a Choi state satisfying Definition~\ref{Def:ChoiNSResources}.
Here, we have reordered Hilbert spaces on $J$ and other Choi states to fit the content of this section.
We write $\Bcal_1^n =\Bcal_1 \otimes \cdots \otimes \Bcal_n$ for the tensor product of $n$ copies of spaces isomorphic to $\Bcal$, and similarly for $\Ycal_1^n$, and we write $\Bcal_1^n \otimes \Ycal_1^n = (\Bcal\Ycal)_1^n$.
This is to be contrasted with $\Bcalov$ defined in Proposition~\ref{Prop:SDPRepresentable}, where the tensor factors were not necessarily isomorphic.
Below, we will consider operators of the form
$\overline{J}_{\Acal\Xcal (\Bcal\Ycal)_1^n}  \in \mathsf{D} (\Acal\Xcal\Bcal_1^n \Ycal_1^n)$.

We write the symmetric group of degree $n$ as $\mathcal{S}_n$, which we identify with the set of bijections $\{1, \ldots, n\} \to \{1, \ldots, n\}$.
Consider the unitary representation
\[
  \Pi: \begin{cases} \mathcal{S}_n & \to \mathcal{U}(\Bcal_1^n\Ycal_1^n) \\ s & \mapsto \Pi_s \end{cases}
\]
of this group, where, for $ \ket{j_1 ... j_n k_1 ... k_n} \in \mathbbm{C}^{\dop{\Bcal_1}} \otimes ... \mathbbm{C}^{\dop{\Bcal_n}} \otimes \mathbbm{C}^{\dop{\Ycal_1}} \otimes ... \mathbbm{C}^{\dop{\Ycal_n}}$:
\begin{equation}
  \label{Eq:PartyPermutation}
  \Pi_s \ket{j_1 ... j_n k_1 ... k_n} =  \ket{j_{s(1)} ... j_{s(n)} k_{s(1)} ... k_{s(n)}}.
\end{equation}

We say that an operator
$\overline{J}_{\Acal\Xcal (\Bcal\Ycal)_1^n}$ is {\em symmetric} if
\begin{equation}
  \label{Eq:ExtensionIsSymmetric}
  \Pi_s \overline{J}_{\Acal\Xcal (\Bcal\Ycal)_1^n} \Pi_s^\dagger = \overline{J}_{\Acal\Xcal (\Bcal\Ycal)_1^n}, \quad \forall s \in \mathcal{S}_n.
\end{equation}

We are now ready to define the set of Choi states with symmetric extensions, where each such Choi state is in one-to-one correspondence with a resource that has a symmetric extension (as defined in the main text).

\begin{definition}
  \label{Def:ChoiStateSymExt}
  The Choi state $J \in \mathsf{D}(\Acal\Xcal\Bcal\Ycal)$ has an $n$-symmetric extension if there exists a state $\overline{J}  \in \mathsf{D} (\Acal\Xcal\Bcal_1^n \Ycal_1^n)$ which obeys the following conditions:
  \begin{enumerate}
  \item The state $\overline{J}$ corresponds to a nonsignaling resource as in Definition~\ref{Def:ChoiNSResources}.
  \item The reduced state
    \[
      \overline{J}_{\Acal\Bcal_1\Xcal\Ycal_1} = \tr_{\Bcal_2 \ldots \Bcal_n \Ycal_2 \ldots \Ycal_n} \overline{J}_{\Acal\Xcal (\Bcal\Ycal)_1^n}
    \]
    is equal to $J_{\Acal\Bcal\Xcal\Ycal}$ itself.
  \item The state $\overline{J}$ is invariant under any joint permutation (in the sense of Eq.~\eqref{Eq:PartyPermutation}) of the $n$ parties of $\Bcal_1^n$ and $\Ycal_1^n$.
  \end{enumerate}
  We denote the set of Choi states having an $n$-symmetric extension as  $\Jbf^{(n)}_{\Acal\Bcal|\Xcal\Ycal} \subseteq \Jbf_{\Acal\Bcal|\Xcal\Ycal}$.
\end{definition}
\noindent Note that if one considers a Choi state of a resource with systems of a special type (trivial or classical), then the copied systems in the extended Choi state will (by symmetry) also be of that special type.
Hence, any constraints from the cases when $\Bcal$ or $\Ycal$ are classical will automatically imply that the copies of these systems satisfy the same constraints.

The set $\Jbf^{(n)}_{\Acal\Bcal|\Xcal\Ycal}$ is SDP-representable, as it is defined as the set of nonsignaling Choi states which satisfy additional {\em linear} constraints.

\subsubsection*{Constructing the hierarchy of SDPs}

As stated in the main text, any LOSR-free resource has an $n$-symmetric extension for all $n$, since it has a decomposition according to Eq.~\eqref{Eq:RConvexDecomposition} of the main text, which can be extended to an arbitrary number of copies of $\mathcal{B}|\mathcal{Y}$ as in~\eqref{Eq:ExplicitSymmetricExtension} of the main text.
Hence, the Choi state $J_{\mathcal{AXBY}}$ corresponding to such a resource has an $n$-symmetric extension for all $n$ as well.

It is not difficult to see that the set of Choi states with $n$-symmetric extensions includes the set of states with $(n-1)$-symmetric extensions.
Thus, we have, mirroring~\eqref{Eq:ResourceHierarchy} of the main text:
\begin{equation}
\Jbf_{\Acal\Bcal|\Xcal\Ycal} \supseteq  \Jbf^{(1)}_{\Acal\Bcal|\Xcal\Ycal} \supseteq \ldots  \Jbf^{(n)}_{\Acal\Bcal|\Xcal\Ycal} \ldots \supseteq \Jbf^\text{free}_{\Acal\Bcal|\Xcal\Ycal}\;.
\end{equation}

Because each set in the hierarchy is SDP-representable, this constitutes a series of tests which is guaranteed to (at some finite level) witness the fact that any given nonfree resource (of any type) is indeed nonfree.

\subsection{Convergence of the SDP hierarchy}

To prove that this sequence converges on the set $\Jbf^\text{free}_{\Acal\Bcal|\Xcal\Ycal}$, we use Theorem~3.4 of Ref.~\cite{Berta2018}.
For convenience, we rewrite the elements of Definition~\ref{Def:ChoiStateSymExt} in the notation of~\cite{Berta2018}, before using their theorem to prove convergence.

The no-signaling condition can be given more explicitly as a pair of constraints: no-signaling from $\Xcal$ to any other party (as in Eq.~\eqref{Eq:RNonsignalingBXY} of the main text) and no-signaling from $\Ycal_n$ to any other party.
By the invariance under joint permutation of parties, this last condition guarantees that for any $n$-symmetric extension, none of the $n$ input systems of $\Bcal_1^n$ can signal to any output of any other party.

To write these in the notation of Ref.~\cite{Berta2018}, we first denote the trace operation as a linear operator, as follows. For generic operators $\Omega_{\Acal\Xcal} \in \mathsf{H}(\Acal\Xcal)$ and $\Omega_{\Bcal_n \Ycal_n} \in \mathsf{H}(\Bcal_n \Ycal_n)$, we define
\begin{equation}
  \Lambda_{\Acal\Xcal \rightarrow \Xcal}(\Omega_{\Acal\Xcal}) = \tr_{\Acal} [\Omega_{\Acal\Xcal}]
\end{equation}
and
\begin{equation}
\Gamma_{\Bcal_n \Ycal_n \rightarrow \Ycal_n}[\Omega_{\Bcal_n \Ycal_n}] = \tr_{\Bcal_n}[\Omega_{\Bcal_n \Ycal_n}].
\end{equation}
Then the constraint of no-signaling from $\mathcal{A}$ to $\mathcal{X}$ is then expressed as
\begin{equation}
  \label{Eq:ANonsignaling}
  \Lambda_{\Acal\Xcal \rightarrow \Xcal} [\overline{J}_{\Acal\Xcal (\Bcal\Ycal)_1^n}] = X_{\Xcal} \otimes \overline{J}_{(\Bcal\Ycal)_1^n},
\end{equation}
where $X_{\Xcal} = \one_{\Xcal} / \dop{\Xcal}$.
The constraint of no-signaling from $\Ycal_n$ to $\Bcal_n$ is then expressed as
\begin{equation}
  \label{Eq:BNonsignaling}
  \Gamma_{\Bcal_n \Ycal_n \rightarrow \Ycal_n}[\overline{J}_{\Acal\Xcal (\Bcal\Ycal)_1^n}] = \overline{J}_{\Acal\Xcal(\Bcal\Ycal)_1^{n - 1}} \otimes Y_{\Ycal_n},
\end{equation}
where $Y_{\Ycal_n} = \one_{\Ycal_n} / \dop{\Ycal_n}$.

Finally, we reexpress the linear constraints from classicality in the notation of Ref.~\cite{Berta2018}.
For a generic operator on an output space, e.g. $\Acal$, we define the classicality operator by
\begin{equation}
  \Xi_{\Acal}(\Omega_\Acal) = \Omega_\Acal - \sum_{j} \matrixel{j}{\Omega_\Acal}{j} \dyad{j}{j}
\end{equation}
for $\Omega_\Acal \in \mathsf{H}(\Acal)$.
For a generic operator on an input space, e.g. $\Xcal$, we define the classicality operator by
\begin{equation}
  \Xi_{\Xcal}(\Omega_\Xcal) = \Omega_\Xcal - \sum_{k} \matrixel{k}{\Omega_\Xcal}{k} \dyad{k}{k}
\end{equation}
 for $\Omega_\Xcal \in \mathsf{H}(\Xcal)$.
Classicality of $\Acal$ and $\Xcal$ (respectively) are then expressed as
\begin{align}
  \label{Eq:AClassical}
 \Xi_{\Acal} &[\overline{J}_{\Acal\Xcal (\Bcal\Ycal)_1^n}] = 0, \\
  \Xi_{\Xcal} &[\overline{J}_{\Acal\Xcal (\Xcal\Ycal)_1^n}] = 0,
\end{align}
while classicality of $\Bcal_n$ and $\Ycal_n$ (respectively) are expressed as
\begin{align}
  \label{Eq:BClassical}
 \Xi_{\Bcal_n} &[\overline{J}_{\Acal\Xcal(\Bcal\Ycal)_1^n}] = 0,\\
  \Xi_{\Ycal_n} &[\overline{J}_{\Acal\Xcal(\Bcal\Ycal)_1^n}] = 0.
\end{align}
Because $\overline{J}_{\Acal\Xcal(\Bcal\Ycal)_1^n}$ is symmetric under exchange of parties, these latter two constraints imply that all copies of $\Bcal$ and $\Ycal$ satisfy the analogous classicality constraints.

\subsubsection*{Convergence of the hierarchy}

Finally, we can leverage Theorem~3.4 of Ref.~\cite{Berta2018} to prove that if there exists a symmetric extension of a Choi state for all $n$, then that Choi state corresponds to a free resource.
In other words, the hierarchy just given converges on the set of Choi states associated to free resources.
It is worth noting that Theorem~3.4 of Ref.~\cite{Berta2018} follows from an application of the de Finetti theorem for quantum channels, which roughly states that any element of an exchangeable sequence of quantum channels can be written as a mixture of many copies of a single unknown channel.)

\begin{proposition}
  If a given Choi state $J_{\Acal\Xcal\Bcal\Ycal} \in \mathsf{D}(\Acal\Xcal\Bcal\Ycal)$ admits an $n$-symmetric extension, then there exists a free resource with Choi state
  $J^{\text{free}}_{\Acal\Xcal\Bcal\Ycal} \in \mathsf{D} (\Acal\Xcal\Bcal\Ycal)$ such that
  \begin{equation}
    \label{thmstatement}
    \left\| J_{\Acal\Xcal\Bcal\Ycal} - J^{\text{free}}_{\Acal\Xcal\Bcal\Ycal} \right\|_1 \leqslant c \sqrt{\frac{1}{n}} ,
  \end{equation}
  where $c$ is a constant that depends on the dimensions $\dop{\Acal}$, $\dop{\Bcal}$, $\dop{\Xcal}$, $\dop{\Ycal}$.
  Furthermore, if the Choi state $J_{\Acal\Xcal\Bcal\Ycal}$ satisfies additional constraints of the form of Eqs.~\eqref{Eq:ClassicalOutput} and \eqref{Eq:ClassicalInput} (from the main text) due to the specific types of the input and output systems, then it follows that $J^{\text{free}}_{\Acal\Xcal\Bcal\Ycal}$ satisfies the same constraints (and hence is of the appropriate type).
\end{proposition}

\begin{proof}
Suppose one has a resource of a specified type whose Choi state $J_{\Acal\Xcal\Bcal\Ycal}$ admits of an $n$-symmetric extension.
Then Theorem~3.4 of~{\cite{Berta2018}} states that there exists a probability distribution $\{ p_i \}$ and operators $\sigma_{\Acal\Xcal}^i$ and $\omega_{\Bcal\Ycal}^i$ such that the operator
  \begin{equation}
    J^{\text{free}}_{\Acal\Xcal\Bcal\Ycal} = \sum_i p_i (\sigma_{\Acal\Xcal}^i \otimes \omega_{\Bcal\Ycal}^i)
  \end{equation}
satisfies the bound of Eq.~\eqref{thmstatement}. It furthermore guarantees that the operators $\sigma_{\Acal\Xcal}^i$ and $\omega_{\Bcal\Ycal}^i$ satisfy
    \begin{equation}
      \Lambda_{\Acal\Xcal \rightarrow \Xcal} [\sigma_{\Acal\Xcal}^i] = X_{\Xcal} = \one_{\Xcal} / \dop{\Xcal},
  \end{equation}
  \begin{equation}
    \Gamma_{\Bcal\Ycal \rightarrow \Ycal} [\omega_{\Bcal\Ycal}^i] = Y_{\Ycal} = \one_{\Ycal_n} / \dop{\Ycal_n},
  \end{equation}
so $\sigma_{\Acal\Xcal}^i$ and $\omega_{\Bcal\Ycal}^i$ are Choi states of single party
  channels, and hence  $J^{\text{free}}_{\Acal\Xcal\Bcal\Ycal}$ is the Choi state of some free resource.
  The only sense in which the specified type of one's resource impacts this argument is that additional classicality constraints of the form in Eqs.~\eqref{Eq:AClassical} or~\eqref{Eq:BClassical} must be included for each classical input or output system.
  These constraints can be embedded in the maps $\Gamma$ and $\Lambda$ of Ref.~{\cite{Berta2018}}.
  The free resource whose existence is guaranteed by the theorem will necessarily satisfy these same constraints, and hence be of the appropriate type.
\end{proof}

If for a given $J_{\Acal\Xcal\Bcal\Ycal}$ there exists an $n$-symmetric extension for all $n$, then the resource it represents is LOSR-free, since the distance to the free set goes to zero in Eq.~\eqref{thmstatement} as $n \to \infty$.

The hierarchy above can be tightened by introducing PPT cuts~{\cite{Doherty2004,Navascues2009}}.
The additional discriminating power granted by such cuts remains to be explored.

\subsection{Special cases where the hierarchy converges after finitely many steps}

In certain cases, the hierarchy converges after a finite number of steps. Here, we generalize an observation made by Terhal et al. in Ref.~\cite{Terhal2003}.

\begin{figure}[b]
  \includegraphics{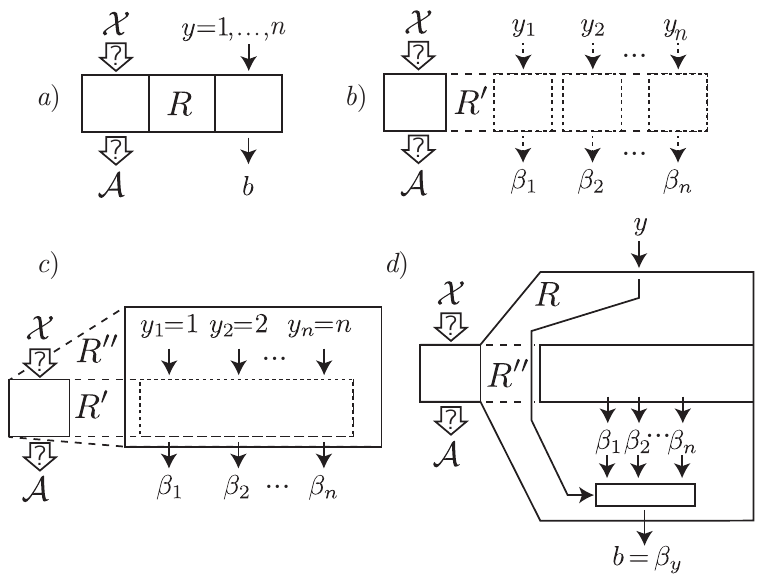}
  \caption{
    \label{Fig:FiniteExtension}
Proof that the hierarchy converges for finite $n$ for certain resource types. See text for details.
}
\end{figure}

\begin{proposition}
  \label{Prop:FiniteConvergence}
A resource of type $\Top{\Xcal}\Csf \to \Top{\Acal}\Csf$ whose output system $\Ycal$ has cardinality $\dop{\Ycal}$ has an $n$-symmetric extension for $n = \dop{\Ycal}$ if and only if it is LOSR-free.
\end{proposition}

\begin{proof}
  As we have already argued, any free resource has an $n$-symmetric extension for all $n$, and hence has one for $n = \dop{\Ycal}$.
  We now prove the converse.
  Let $R$ be a resource of the type considered in the proposition, as shown in Fig.~\ref{Fig:FiniteExtension}a.
  Assume now that $R$ has an $n$-symmetric extension with $n = \dop{\Ycal}$, as shown in Fig.~\ref{Fig:FiniteExtension}b, with inputs $y_1$ to $y_n$ and outputs $\beta_1$ to $\beta_n$.
  Next, imagine that Alice and Bob share such a resource, where Alice has possession of $\Xcal$ and $\Acal$, while Bob has possession of all the other input and output systems; that is, Bob's (classical) input is $\overline{y} = (y_1, \ldots, y_n)$, and his (classical) output is $\overline{\beta} = (\beta_1, \ldots, \beta_n)$.
  In Fig.~\ref{Fig:FiniteExtension}c, we consider a local (and hence free) transformation that Bob can implement, wherein he inputs the classical values $y_1 = 1$, \ldots $y_n = n$, obtaining as his outcome a tuple of values for $\overline{\beta}$. The resulting resource $R''$ has type $\Top{\Xcal}\Isf \to \Top{\Acal}\Csf$, which is a useless type according to Prop.~\ref{Prop:Useless}.
  However, $R''$ can be converted back to $R$ using an LOSR-free transformation, as shown in Fig.~\ref{Fig:FiniteExtension}; namely, Bob just selects the value of the $y$th output ($\beta_y$), for whatever input $y$ he is given.
  Since this free transformation preserves the set of LOSR-free resources and yet generates $R$ from a free resource, $R$ itself must be free.
\end{proof}

Note that these are cases for which an exact representation exists, as described in Section~\ref{App:ExactFreeRepresentable}.

\section{Representability of monotones}
\label{App:MonotoneRepresentable}

Our hierarchy can be used to approximate the value of a monotone on any given resource, as we now illustrate using the type-independent absolute robustness as an example.

We start by rewriting the computation of the absolute robustness (Definition~\ref{Def:AbsoluteRobustness} of the main text) using the Choi representation:
\begin{multline}
  M_{\text{abs}} (R_{\Acal\Bcal|\Xcal\Ycal}) = \min s \ \ \text{s.t. } \ \frac{J_R + s J_S}{1 + s} \in \Jbf^{\rm free}_{\Acal\Bcal|\Xcal\Ycal}\;,\\
  s \geqslant 0, \qquad J_S \in \Jbf^{\rm free}_{\Acal\Bcal|\Xcal\Ycal}\;.
\end{multline}
To practically compute $ M_{\text{abs}} (R_{\Acal\Bcal|\Xcal\Ycal})$, one faces two problems.
First, the element $s J_S$ is not SDP-representable as it is a product of two variables being optimized over.
However, we can rewrite the definition as follows:
\begin{multline}
  \label{Eq:AbsRobConic}
  M_{\text{abs}} (R_{\Acal\Bcal|\Xcal\Ycal}) = \min \tr(\hat{J}_S) \ \ \text{s.t. } \ J_R + \hat{J}_S \in \hat{\Jbf}_{\Acal\Bcal|\Xcal\Ycal}^{\rm free}\;, \\
  \qquad \hat{J}_S \in \hat{\Jbf}_{\Acal\Bcal|\Xcal\Ycal}^{\rm free}\;,
\end{multline}
where we defined $\hat{J}_S = s J_S$, and thus $\tr(\hat{J}_S) = s$, and we eliminated the denominator $(1+s)$ by utilizing the conic extension  $\hat{\Jbf}_{\Acal\Bcal|\Xcal\Ycal}^{\rm free}$ (according to Definition~\ref{Def:ConicExtension}).
Second, the set $\Jbf^{\rm free}_{\Acal\Bcal|\Xcal\Ycal}$ is not necessarily LP- or SDP-representable.
This problem can be addressed by replacing the instances of $\Jbf_{\Acal\Bcal|\Xcal\Ycal}^{\rm free}$ by the SDP-representable outer approximation $\Jbf_{\Acal\Bcal|\Xcal\Ycal}^{(n)}$ for some integer $n$.
Since this means one is minimizing over a strictly larger set, the value computed in this way constitutes a lower bound $\tilde{M}_{\text{abs}}^{(n)}(R_{\Acal\Bcal|\Xcal\Ycal}) \le M_\text{abs}(R_{\Acal\Bcal|\Xcal\Ycal})$.
That is, one can compute
\begin{multline}
  \tilde{M}^{(n)}_{\text{abs}} (R_{\Acal\Bcal|\Xcal\Ycal}) = \min \tr(\hat{J}_S) \ \ \text{s.t. } \ J_R + \hat{J}_S \in \hat{\Jbf}^{(n)}_{\Acal\Bcal|\Xcal\Ycal}\;,\\
  \qquad \hat{J}_S \in \hat{\Jbf}^{(n)}_{\Acal\Bcal|\Xcal\Ycal}\;,
\end{multline}
giving a sequence of increasingly-tight lower bounds on the desired value $\tilde{M}_{\text{abs}} (R)$, where each lower bound is computable by an SDP, and such that the sequence converges on $\tilde{M}_{\text{abs}} (R)$.

Other monotones, such as the generalized robustness~\eqref{Eq:GeneralizedRobustness} and the nonlocal weight~\eqref{Eq:NonlocalWeight}, can be approximated in a similar manner.

\begin{figure}[b]
  \includegraphics{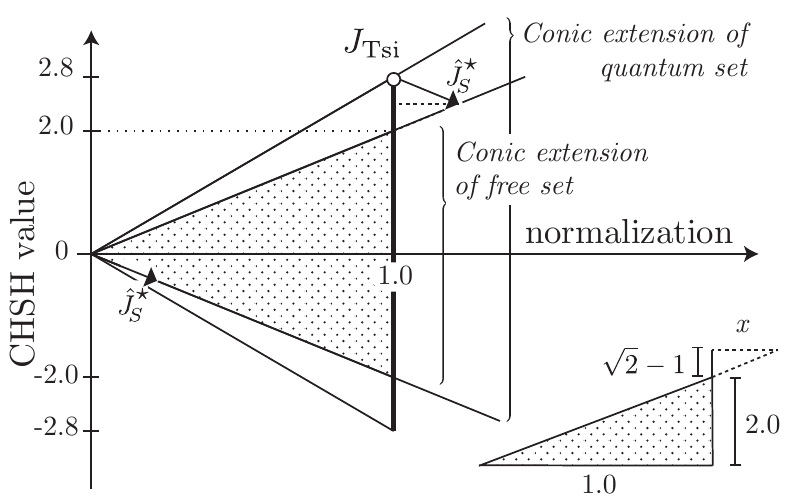}
  \caption{
    \label{Fig:AbsRobTsirelson}
  Graphical argument for computing the absolute robustness of the Tsirelson box. See text for details.
}
\end{figure}

\section{Computation of the values of $M_\text{abs}$ for the five example resources in Fig.~\ref{Fig:AbsRobPreorder} of the main text}
\label{App:AbsoluteRobustnessComputation}
Since our type-independent monotone reduces to that of Ref.~\cite{Vidal1999} for the specific case of quantum states, the absolute robustness of the singlet can be directly computed from its Schmidt coefficients using the formulas therein:
\begin{equation}
  \label{Eq:AbsoluteRobustnessPureState}
  M_\text{abs}\left( \sum_{i=1}^d \lambda_i \ket{ii} \right) = \left(\sum_{i=1}^d \lambda_i \right)^2 - 1,
\end{equation}
where, implicitly, $M_\text{abs}(\cdot)$ applies on the density matrix derived from the pure state with the given Schmid decomposition.
We also verified this value using the PPT criterion (Section~\ref{App:FreeRepresentable}).
The absolute robustness of the distributed measurement is equal to that of the singlet, as we previously proved that these resources are interconvertible under LOSR operations.

The values of the absolute robustness for the two steering assemblage examples were computed using the exact SDP representation in Section~\ref{App:FreeRepresentable}.

The values of the absolute robustness for the Tsirelson box were computed using the exact LP representation given in Section~\ref{App:FreeRepresentable}.

\subsection{A graphical understanding of absolute robustness}
We also verified this value by the following graphical argument, which corresponds to Fig.~\ref{Fig:AbsRobTsirelson}.

In the CHSH scenario with binary inputs and outputs, we define the {\em isotropic line} as the one-dimensional affine subspace that contains the Tsirelson box and the uniformly random distribution.
By extension this isotropic line also contains the PR-box.
This isotropic line corresponds to the bold line of the figure, at the normalization coordinate (abcissa) $=1.0$.
The part of the isotropic line on the CHSH value (ordinate) interval $[-2\sqrt{2}, 2\sqrt{2}]$ can be realized with quantum resources, and thus represents the quantum set.
The part of the isotropic with CHSH value interval $[-2,2]$ corresponds to the free set.
Any probability distribution can be projected on the isotropic line by symmetrization under the symmetries of the CHSH inequality ($I=\sum_{abxy} (-1)^{a+b+xy} P(ab|xy)$), and that projection does not affect the CHSH expectation value.
Thus, the search for the optimal noise to mix can be restricted to the conic extension of that subspace.
Using the definition~\eqref{Eq:AbsRobConic} of the absolute robustness, we draw the conic extensions of both the quantum realizable set and the free set; and we dot the area of the latter where it corresponds to subnormalized resources.

The Tsirelson box $J_\text{Tsi}$ can be brought into the conic extension of the free set by adding to it a subnormalized free resource $\hat{J}_S$, and we can find an optimal solution of the minimization~\eqref{Eq:AbsRobConic} in this dotted area.
We show the optimal solution graphically as $\hat{J}_S^\star$, and denote its trace by $x$.
Using the intercept theorem~\cite{Ostermann2012}, we obtain $M_\text{abs}(J_\text{Tsi}) = x = (\sqrt{2}-1)/2 \approx 0.207$.

\subsection{Extension to partially entangled states}

We now consider parameterized variants of several of the resources studied in the main text.

Consider first a generic two-qubit state
\begin{equation} \label{noisystate}
  \ket{\psi_\alpha} = \cos \alpha \ket{01} - \sin \alpha \ket{10},
\end{equation}
where $\alpha \in [0, \pi/4]$; clearly, for $\alpha = \pi/4$ one recovers the singlet state.
From Eq.~\eqref{Eq:AbsoluteRobustnessPureState}, one has $M_\text{abs}(\ketbra{\psi_\alpha}{\psi_\alpha}) = \sin 2 \alpha$, as plotted in Figure~\ref{Fig:Partial}.

Next, consider the following:
\begin{itemize}
\item An assemblage $\{\mu_{\alpha}^3(a|x)\}$ with three measurement settings $x=1,2,3$ and two measurement outcomes $a=1,2$;
\item An assemblage $\{\mu_{\alpha}^2(a|x)\}$ with two measurement settings $x=1,2$ and two measurement outcomes $a=1,2$;
\item A box $P_\alpha(ab|xy)$ with two measurement settings $x,y=1,2$ and two measurement outcomes $a,b=1,2$ for each party.
\end{itemize}
We now define a family of resources of each of these types and dimensionalities, defined by the maximum achievable $M_\text{abs}$ that can be achieved starting from the state in Eq.~\eqref{noisystate}, as one ranges over $\alpha\in[0,\pi/4]$. To compute this family of resources explicitly, we computed the projective measurement(s) which generate a resource of each of these specified types and dimensionalities.

The absolute robustness of each of these families of resources (computed using the same methods as for the specific resources in the main text) is also plotted in Figure~\ref{Fig:Partial}. One can see that the behavior of $M_\text{abs}$ mirrors that of the specific resources in the main text, in that for all $\alpha$, the state resource takes the highest value, followed by the fine-grained assemblage, while the coarse-grained assemblage and the box take the lowest values (which are equal).

Interestingly, the optimal projective measurements can be chosen so that the sequence of conversions
\begin{equation}
  \ket{\psi_\alpha} \rightarrow \mu_\alpha^3 \rightarrow \mu_\alpha^2 \rightarrow P_\alpha
\end{equation}
are possible between those optimal objects, analogous to conversion relations holding among the specific resources in the main text.

Namely, for any $\alpha\in[0,\pi/4]$, the measurements to obtain the optimal $\mu^3_\alpha$ can be chosen to be in the X basis ($\ket{\pm}$), the Y basis ($\ket{\pm i}$) and the Z basis ($\ket{0}, \ket{1}$).
The measurements to obtain the optimal $\mu^2_\alpha$ can simply be those in the X and the Z basis.
These are the same measurements defining the specific assemblages in the main text.
To transform $\mu^2_\alpha$ into the optimal box $P_\alpha$, however, the measurements are slightly adapted, corresponding to the POVM elements
\begin{equation}
  B_{b|y} = \frac{\mathbbm{1} + (-1)^b \left( \sin \theta_y \sigma_1 + \cos \theta_y \sigma_3 \right)}{2},
\end{equation}
 where $\theta_1 = \pi - \theta_0$.
We obtain $\theta_0= \pi/4$ for $\alpha=\pi/4$, with $\theta_0$ monotonically decreasing towards 0 as $\alpha \rightarrow 0$.

\begin{figure}[htb!]
  \includegraphics{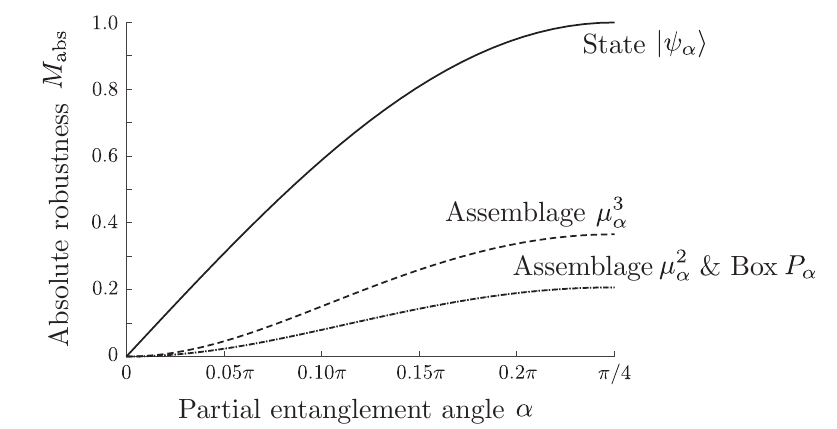}
  \caption{Extension of the examples of the main text to resources derived from partially entangled qubit states.
    \label{Fig:Partial}
}
\end{figure}

\onecolumngrid

\section{Computer code}

We kept the code simple and limited to the cases considered in our study. It runs on Matlab~\cite{MATLAB} or Octave~\cite{Octave520}, and it requires the modeling toolbox YALMIP~\cite{Lofberg2004}, an SDP solver such as SeDuMi~\cite{Sturm1999}, and the quantum information toolbox QETLAB~\cite{Johnston2016}.

We provide below the convenience functions that transform states into various resources, the functions that compute the absolute robustness for the examples presented, and finally the script that computes the graphs in Figure~\ref{Fig:Partial}.
\subsection{Convenience functions transforming resources}
\begin{lstlisting}[breaklines]
function sigma = state_to_asm2(dm, x)
% Computes the assemblage obtained by two projective measurements on a qubit-qubit state
%
% The measurements are given by cos(theta/2) |0> + sin(theta/2)*exp(1i*phi) |1>
%
% The assemblage is returned as a 4D array; sigma(:,:,a,x) is the result
% for the measurement outcome a and the setting x.
%
% Args:
%   dm (double(4,4)): Qubit-qubit density matrix
%   x (double(4,1)): Angles of 2 projective measurements, given as [theta1;ph1i;theta2;phi2]
%
% Returns:
%   double(2,2,2,2): Computed assemblage
    A = zeros(2,2,2);
    sigma = zeros(2,2,2,2);
    for i = 1:2
        theta = x((i-1)*2+1);
        phi = x((i-1)*2+2);
        A(:,:,i) = cos(phi)*sin(theta)*[0 1; 1 0] + sin(phi)*sin(theta)*[0 -1i; 1i 0] + ...
            cos(theta)*[1 0; 0 -1];
        sigma(:,:,1,i) = PartialTrace(dm*kron((eye(2)+A(:,:,i))/2, eye(2)), 1, [2 2]);
        sigma(:,:,2,i) = PartialTrace(dm*kron((eye(2)-A(:,:,i))/2, eye(2)), 1, [2 2]);
    end
end
\end{lstlisting}
\begin{lstlisting}[breaklines]
function sigma = state_to_asm3(dm, x)
% Computes the assemblage obtained by three projective measurements on a qubit-qubit state
%
% The measurements are given by cos(theta/2) |0> + sin(theta/2)*exp(1i*phi) |1>
%
% The assemblage is returned as a 4D array such that sigma(:,:,a,x) is the result for the
% measurement outcomes a and the setting x.
%
% Args:
%   dm (double(4,4)): Qubit-qubit density matrix
%   x (double(6,1)): Angles of 3 projective measurements,
%                    given as [theta1;ph1i;theta2;phi2;theta3;phi3]
%
% Returns:
%   double(2,2,2,2): Computed assemblage
    A = zeros(2,2,3);
    sigma = zeros(2,2,2,3);
    for i = 1:3
        theta = x((i-1)*2+1);
        phi = x((i-1)*2+2);
        A(:,:,i) = cos(phi)*sin(theta)*[0 1; 1 0] + sin(phi)*sin(theta)*[0 -1i; 1i 0] + ...
            cos(theta)*[1 0; 0 -1];
        sigma(:,:,1,i) = PartialTrace(dm*kron((eye(2)+A(:,:,i))/2, eye(2)), 1, [2 2]);
        sigma(:,:,2,i) = PartialTrace(dm*kron((eye(2)-A(:,:,i))/2, eye(2)), 1, [2 2]);
    end
end
\end{lstlisting}
\begin{lstlisting}[breaklines]
function P = state_to_box(dm, angles)
% Computes the correlations obtained by two projective measurements on a qubit-qubit state
%
% The measurements are given by cos(theta/2) |0> + sin(theta/2)*exp(1i*phi) |1>
%
% The correlations are returned as a 4D array such that P(a,b,x,y) is the result
% for the measurement xoutcomes a,b and the settings x,y.
%
% Args:
%   dm (double(4,4)): Qubit-qubit density matrix
%   angles (double(8,1)): Angles of 4 projective measurements, given as
%                         [A_theta1;A_ph1i;A_theta2;A_phi2;B_theta1;B_phi1;B_theta2;B_phi2]
% Returns:
%   double(2,2,2,2): Correlations P(a,b,x,y)
    A = zeros(2,2,2);
    B = zeros(2,2,2);
    P = zeros(2,2,2,2);
    for x = 1:2
        At = angles((x-1)*2+1);
        Ap = angles((x-1)*2+2);
        A(:,:,x) = cos(Ap)*sin(At)*[0 1; 1 0] + sin(Ap)*sin(At)*[0 -1i; 1i 0] + ...
            cos(At)*[1 0; 0 -1];
    end
    for y = 1:2
        Bt = angles(4+(y-1)*2+1);
        Bp = angles(4+(y-1)*2+2);
        B(:,:,y) = cos(Bp)*sin(Bt)*[0 1; 1 0] + sin(Bp)*sin(Bt)*[0 -1i; 1i 0] + ...
            cos(Bt)*[1 0; 0 -1];
    end
    for x = 1:2
        for y = 1:2
            P(1,1,x,y) = trace(kron(eye(2)+A(:,:,x), eye(2)+B(:,:,y))*dm)/4;
            P(2,1,x,y) = trace(kron(eye(2)-A(:,:,x), eye(2)+B(:,:,y))*dm)/4;
            P(1,2,x,y) = trace(kron(eye(2)+A(:,:,x), eye(2)-B(:,:,y))*dm)/4;
            P(2,2,x,y) = trace(kron(eye(2)-A(:,:,x), eye(2)-B(:,:,y))*dm)/4;
        end
    end
end
\end{lstlisting}

\subsection{Functions computing the absolute robustness for particular types}
\begin{lstlisting}[breaklines]
function ar = ar_state(dm)
% Computes the absolute robustness of the given qubit-qubit state
%
% Uses the PPT criterion to detect separability (sufficient and necessary here)
%
% Needs QETLAB, YALMIP and a SDP solver in the path
%
% Args:
%   dm (double(4,4)): Qubit-qubit density matrix
%
% Returns:
%   double: Absolute robustness of the state
    mix = sdpvar(4,4,'hermitian','complex');
    res = dm + mix;
    F = [mix >= 0
         res >= 0
         PartialTranspose(mix, 2, [2 2]) >= 0
         PartialTranspose(res, 2, [2 2]) >= 0];
    optimize(F, trace(mix), sdpsettings('verbose', 0));
    ar = double(trace(mix));
end
\end{lstlisting}
\begin{lstlisting}[breaklines]
function ar = ar_asm2(sigma)
% Computes the absolute robustness of the given steering assemblage
%
% This is implemented for binary outcomes and settings for the classical party,
% qubit for the quantum party
%
% Needs YALMIP and a SDP solver in the path
%
% Args:
%   sigma (double(2,2,2,2)): sigma(i,j,a,x) = <j| sigma_{a|x} |i>, i,j=1,2,
%                            a=1,2 is the measurement outcome
%                            x=1,2 the measurement setting
%
% Returns:
%   double: Absolute robustness of the assemblage
    lhs1 = sdpvar(2,2,2,2,'hermitian','complex');
    rho1 = sdpvar(2,2,2,2,'hermitian','complex');
    % the local hidden variable is a tuple (alpha1, alpha2)
    % representing a deterministic probability distribution for Alice
    % alpha1 is the value of the outcome a for x = 1, and
    % alpha2 is the value of the outcome a for x = 2
    %
    % The quantum state on the Bob subsystem is lhs1(i,j,alpha1,alpha2),
    % from which we compute the assemblage rho1(i,j,a,x)
    rho1(:,:,1,1) = lhs1(:,:,1,1) + lhs1(:,:,1,2);
    rho1(:,:,2,1) = lhs1(:,:,2,1) + lhs1(:,:,2,2);
    rho1(:,:,1,2) = lhs1(:,:,1,1) + lhs1(:,:,2,1);
    rho1(:,:,2,2) = lhs1(:,:,1,2) + lhs1(:,:,2,2);

    % same story for lhs2, rho2
    lhs2 = sdpvar(2,2,2,2,'hermitian','complex');
    rho2 = sdpvar(2,2,2,2,'hermitian','complex');
    rho2(:,:,1,1) = lhs2(:,:,1,1) + lhs2(:,:,1,2);
    rho2(:,:,2,1) = lhs2(:,:,2,1) + lhs2(:,:,2,2);
    rho2(:,:,1,2) = lhs2(:,:,1,1) + lhs2(:,:,2,1);
    rho2(:,:,2,2) = lhs2(:,:,1,2) + lhs2(:,:,2,2);

    % assemblage elements are SDP
    F1 = [lhs1(:,:,1,1) >= 0; lhs1(:,:,2,1) >= 0; lhs1(:,:,1,2) >= 0; lhs1(:,:,2,2) >= 0];
    F2 = [lhs2(:,:,1,1) >= 0; lhs2(:,:,2,1) >= 0; lhs2(:,:,1,2) >= 0; lhs2(:,:,2,2) >= 0];

    % constraint of the SDP-representation of absolute robustness
    F3 = [sigma + rho1 == rho2];

    % the objective is the normalization factor of rho1,
    % which we can compute on any input (here wlog for x = 1)
    obj = trace(rho1(:,:,1,1) + rho1(:,:,2,1));

    % perform the SDP optimization
    optimize([F1;F2;F3], obj, sdpsettings('verbose', 0));

    % return the numerical solution
    ar = double(obj);
end
\end{lstlisting}
\begin{lstlisting}[breaklines]
function ar = ar_asm3(sigma)
% Computes the absolute robustness of the given steering assemblage (3 settings type)
%
% This is implemented for:
% - binary outcomes, three settings for the classical party (Alice)
% - qubit for the quantum party (Bob)
%
% Needs YALMIP and a SDP solver in the path
%
% See ar_asm2.m for more detailed comments
%
% Args:
%   sigma (double(2,2,2,2)): sigma(i,j,a,x) = <j| sigma_{a|x} |i>, i,j=1,2,
%                            a=1,2 is the measurement outcome
%                            x=1,2,3 the measurement setting
%
% Returns:
%   double: Absolute robustness of the assemblage
    lhs1 = sdpvar(2,2,2,2,2,'hermitian','complex');
    rho1 = sdpvar(2,2,2,3,'hermitian','complex');
    % The local hidden variable is the triplet (alpha1, alpha2, alpha3)
    % where alpha1=a(x=1) alpha2=a(x=2) alpha3=a(x=3)
    % LHS model is given by lhs1(:,:,alpha1,alpha2,alpha3)
    rho1(:,:,1,1) = lhs1(:,:,1,1,1) + lhs1(:,:,1,2,1) + lhs1(:,:,1,1,2) + lhs1(:,:,1,2,2);
    rho1(:,:,2,1) = lhs1(:,:,2,1,1) + lhs1(:,:,2,2,1) + lhs1(:,:,2,1,2) + lhs1(:,:,2,2,2);
    rho1(:,:,1,2) = lhs1(:,:,1,1,1) + lhs1(:,:,2,1,1) + lhs1(:,:,1,1,2) + lhs1(:,:,2,1,2);
    rho1(:,:,2,2) = lhs1(:,:,1,2,1) + lhs1(:,:,2,2,1) + lhs1(:,:,1,2,2) + lhs1(:,:,2,2,2);
    rho1(:,:,1,3) = lhs1(:,:,1,1,1) + lhs1(:,:,2,1,1) + lhs1(:,:,1,2,1) + lhs1(:,:,2,2,1);
    rho1(:,:,2,3) = lhs1(:,:,1,1,2) + lhs1(:,:,2,1,2) + lhs1(:,:,1,2,2) + lhs1(:,:,2,2,2);
    lhs2 = sdpvar(2,2,2,2,2,'hermitian','complex');
    rho2 = sdpvar(2,2,2,3,'hermitian','complex');
    % same construction, with rho2 the resulting LOSR-free assemblage
    rho2(:,:,1,1) = lhs2(:,:,1,1,1) + lhs2(:,:,1,2,1) + lhs2(:,:,1,1,2) + lhs2(:,:,1,2,2);
    rho2(:,:,2,1) = lhs2(:,:,2,1,1) + lhs2(:,:,2,2,1) + lhs2(:,:,2,1,2) + lhs2(:,:,2,2,2);
    rho2(:,:,1,2) = lhs2(:,:,1,1,1) + lhs2(:,:,2,1,1) + lhs2(:,:,1,1,2) + lhs2(:,:,2,1,2);
    rho2(:,:,2,2) = lhs2(:,:,1,2,1) + lhs2(:,:,2,2,1) + lhs2(:,:,1,2,2) + lhs2(:,:,2,2,2);
    rho2(:,:,1,3) = lhs2(:,:,1,1,1) + lhs2(:,:,2,1,1) + lhs2(:,:,1,2,1) + lhs2(:,:,2,2,1);
    rho2(:,:,2,3) = lhs2(:,:,1,1,2) + lhs2(:,:,2,1,2) + lhs2(:,:,1,2,2) + lhs2(:,:,2,2,2);
    F1 = [lhs1(:,:,1,1,1) >= 0; lhs1(:,:,2,1,1) >= 0; ...
          lhs1(:,:,1,2,1) >= 0; lhs1(:,:,2,2,1) >= 0; ...
          lhs1(:,:,1,1,2) >= 0; lhs1(:,:,2,1,2) >= 0; ...
          lhs1(:,:,1,2,2) >= 0; lhs1(:,:,2,2,2) >= 0];
    F2 = [lhs2(:,:,1,1,1) >= 0; lhs2(:,:,2,1,1) >= 0; ...
          lhs2(:,:,1,2,1) >= 0; lhs2(:,:,2,2,1) >= 0; ...
          lhs2(:,:,1,1,2) >= 0; lhs2(:,:,2,1,2) >= 0; ...
          lhs2(:,:,1,2,2) >= 0; lhs2(:,:,2,2,2) >= 0];
    F3 = [sigma + rho1 == rho2];
    obj = trace(rho1(:,:,1,1) + rho1(:,:,2,1));
    optimize([F1;F2;F3], obj, sdpsettings('verbose', 0));
    ar = double(obj);
end
\end{lstlisting}
\begin{lstlisting}[breaklines]
function ar = ar_box(P)
% Computes the absolute robustness of a box, two parties, binary I/O
%
% Needs YALMIP and a LP solver in the path
%
% Args:
%   P (double(2,2,2,2)): Correlations with indexing P(a,b,x,y)
    P = permute(P, [1 3 2 4]); % index with P(a,x,b,y) so that we can use the tensor structure
    P = P(:);
    % single party deterministic points, rows are P(1|1), P(2|1), P(1|2), P(2|2),
    % columns are deterministic strategies
    M = [1 0 1 0
         0 1 0 1
         1 1 0 0
         0 0 1 1];
    % deterministic points for two parties with binary I/O, we use the tensor product construction
    M2 = kron(M, M);
    % weights of deterministic strategies
    res_weights = sdpvar(16, 1);
    mix_weights = sdpvar(16, 1);
    res_P = M2 * res_weights;
    mix_P = M2 * mix_weights;
    res = P(:) + mix_P;
    F = [mix_weights >= 0
         res_weights >= 0
         res == res_P];
    % normalization factor is the sum of the weights of deterministic strategies
    obj = sum(mix_weights);
    solvesdp(F, obj, sdpsettings('verbose', 0));
    ar = double(obj);
end
\end{lstlisting}

\subsection{Main script}
\begin{lstlisting}[breaklines]
% Computes the absolute robustness of the various objects presented in the Supp. Mat.
%
% This script requires NLOPT, YALMIP, QETLAB and a SDP solver
%
% https://github.com/stevengj/nlopt
%   (used for the derivative free BOBYQA nonlinear optimization algorithm)
% https://github.com/yalmip/YALMIP
%   (modeling toolbox, used for the LP/SDP characterizations)
% https://github.com/nathanieljohnston/QETLAB
%   (for PartialTrace, PartialTranspose)
% https://github.com/sqlp/sedumi
%   (SDP&LP solver)

angles = 0:pi/128:pi/4;

% computed monotones (abs. robustness)
M_state = [];
M_asm2 = [];
M_asm3 = [];
M_box = [];

% stores the optimal measurement angles for the box, those don't have a closed form solution
thetas = [];

% we know the optimal angle for angles(end) = pi/4
topt = pi/4;
% so we start the iteration from the point for which we know the answer
for i = length(angles):-1:1
    % construct the state and compute its absolute robustness
    a = angles(i);
    ket = [0;cos(a);-sin(a);0];
    dm = ket*ket';
    % absolute robustness of state
    M_state(i) = ar_state(dm);
    % or M_state(i) = (cos(a)+sin(a))^2-1 using the exact formula

    % compute steering assemblages
    asm3 = state_to_asm3(dm, [0;0;pi/2;0;pi/2;pi/2]);
    asm2 = state_to_asm2(dm, [0;0;pi/2;0]);
    % compute absolute robustness of those assemblages
    M_asm3(i) = ar_asm3(asm3);
    M_asm2(i) = ar_asm2(asm2);

    % find the measurement angle which gives maximal absolute robustness for the box
    % note that this is a double optimization:
    % - outer optimization over the measurement angles, which we solve with nlopt/BOBYQA
    % - inner convex optimization to compute the absolute robustness
    opt.algorithm = NLOPT_LN_BOBYQA;
    opt.verbose = 1;
    opt.ftol_rel = 1e-7;
    opt.initial_step = [0.1];
    opt.ftol_abs = 1e-7;
    opt.maxeval = 50;
    opt.max_objective = @(t) ar_box(state_to_box(dm, [0;0;pi/2;0;t;0;pi-t;0]));
    % as a starting point for the optimization, we reuse the optimal point
    % from the previous state considered
    [topt objval] = nlopt_optimize(opt, topt);
    M_box(i) = objval;
    thetas(i) = topt;
end
% make the plot
clf
hold on
plot(angles/pi, M_state, 'k-');
plot(angles/pi, M_asm3, 'k--');
plot(angles/pi, M_asm2, 'k-.');
plot(angles/pi, M_box, 'k:');
legend({'State', 'Assemblage x=1,2,3', 'Assemblage x=1,2', 'Box'})
axis([0 1/4 0 1])
xlabel('Partial entanglement angle \alpha');
ylabel('Absolute robustness M_{abs}');
\end{lstlisting}

\bibliography{refs}

\begin{thebibliography}{67}%
\makeatletter
\providecommand \@ifxundefined [1]{%
 \@ifx{#1\undefined}
}%
\providecommand \@ifnum [1]{%
 \ifnum #1\expandafter \@firstoftwo
 \else \expandafter \@secondoftwo
 \fi
}%
\providecommand \@ifx [1]{%
 \ifx #1\expandafter \@firstoftwo
 \else \expandafter \@secondoftwo
 \fi
}%
\providecommand \natexlab [1]{#1}%
\providecommand \enquote  [1]{``#1''}%
\providecommand \bibnamefont  [1]{#1}%
\providecommand \bibfnamefont [1]{#1}%
\providecommand \citenamefont [1]{#1}%
\providecommand \href@noop [0]{\@secondoftwo}%
\providecommand \href [0]{\begingroup \@sanitize@url \@href}%
\providecommand \@href[1]{\@@startlink{#1}\@@href}%
\providecommand \@@href[1]{\endgroup#1\@@endlink}%
\providecommand \@sanitize@url [0]{\catcode `\\12\catcode `\$12\catcode
  `\&12\catcode `\#12\catcode `\^12\catcode `\_12\catcode `\%12\relax}%
\providecommand \@@startlink[1]{}%
\providecommand \@@endlink[0]{}%
\providecommand \url  [0]{\begingroup\@sanitize@url \@url }%
\providecommand \@url [1]{\endgroup\@href {#1}{\urlprefix }}%
\providecommand \urlprefix  [0]{URL }%
\providecommand \Eprint [0]{\href }%
\providecommand \doibase [0]{http://dx.doi.org/}%
\providecommand \selectlanguage [0]{\@gobble}%
\providecommand \bibinfo  [0]{\@secondoftwo}%
\providecommand \bibfield  [0]{\@secondoftwo}%
\providecommand \translation [1]{[#1]}%
\providecommand \BibitemOpen [0]{}%
\providecommand \bibitemStop [0]{}%
\providecommand \bibitemNoStop [0]{.\EOS\space}%
\providecommand \EOS [0]{\spacefactor3000\relax}%
\providecommand \BibitemShut  [1]{\csname bibitem#1\endcsname}%
\let\auto@bib@innerbib\@empty
\bibitem [{\citenamefont {Horodecki}\ \emph {et~al.}(2009)\citenamefont
  {Horodecki}, \citenamefont {Horodecki}, \citenamefont {Horodecki},\ and\
  \citenamefont {Horodecki}}]{Horodecki2009}%
  \BibitemOpen
  \bibfield  {author} {\bibinfo {author} {\bibfnamefont {R.}~\bibnamefont
  {Horodecki}}, \bibinfo {author} {\bibfnamefont {P.}~\bibnamefont
  {Horodecki}}, \bibinfo {author} {\bibfnamefont {M.}~\bibnamefont
  {Horodecki}}, \ and\ \bibinfo {author} {\bibfnamefont {K.}~\bibnamefont
  {Horodecki}},\ }\href {\doibase 10.1103/RevModPhys.81.865} {\bibfield
  {journal} {\bibinfo  {journal} {Reviews of Modern Physics}\ }\textbf
  {\bibinfo {volume} {81}},\ \bibinfo {pages} {865} (\bibinfo {year}
  {2009})}\BibitemShut {NoStop}%
\bibitem [{\citenamefont {Fine}(1982)}]{Fine1982}%
  \BibitemOpen
  \bibfield  {author} {\bibinfo {author} {\bibfnamefont {A.}~\bibnamefont
  {Fine}},\ }\href {\doibase 10.1103/PhysRevLett.48.291} {\bibfield  {journal}
  {\bibinfo  {journal} {Phys. Rev. Lett.}\ }\textbf {\bibinfo {volume} {48}},\
  \bibinfo {pages} {291} (\bibinfo {year} {1982})}\BibitemShut {NoStop}%
\bibitem [{\citenamefont {Brunner}\ \emph {et~al.}(2014)\citenamefont
  {Brunner}, \citenamefont {Cavalcanti}, \citenamefont {Pironio}, \citenamefont
  {Scarani},\ and\ \citenamefont {Wehner}}]{Brunner2014}%
  \BibitemOpen
  \bibfield  {author} {\bibinfo {author} {\bibfnamefont {N.}~\bibnamefont
  {Brunner}}, \bibinfo {author} {\bibfnamefont {D.}~\bibnamefont {Cavalcanti}},
  \bibinfo {author} {\bibfnamefont {S.}~\bibnamefont {Pironio}}, \bibinfo
  {author} {\bibfnamefont {V.}~\bibnamefont {Scarani}}, \ and\ \bibinfo
  {author} {\bibfnamefont {S.}~\bibnamefont {Wehner}},\ }\href {\doibase
  10.1103/RevModPhys.86.419} {\bibfield  {journal} {\bibinfo  {journal} {Rev.
  Mod. Phys.}\ }\textbf {\bibinfo {volume} {86}},\ \bibinfo {pages} {419}
  (\bibinfo {year} {2014})}\BibitemShut {NoStop}%
\bibitem [{\citenamefont {Cavalcanti}\ and\ \citenamefont
  {Skrzypczyk}(2017)}]{Cavalcanti2017}%
  \BibitemOpen
  \bibfield  {author} {\bibinfo {author} {\bibfnamefont {D.}~\bibnamefont
  {Cavalcanti}}\ and\ \bibinfo {author} {\bibfnamefont {P.}~\bibnamefont
  {Skrzypczyk}},\ }\href {\doibase 10.1088/1361-6633/80/2/024001} {\bibfield
  {journal} {\bibinfo  {journal} {Rep. Prog. Phys.}\ }\textbf {\bibinfo
  {volume} {80}},\ \bibinfo {pages} {024001} (\bibinfo {year}
  {2017})}\BibitemShut {NoStop}%
\bibitem [{\citenamefont {Cavalcanti}\ \emph {et~al.}(2017)\citenamefont
  {Cavalcanti}, \citenamefont {Skrzypczyk},\ and\ \citenamefont {{\v
  S}upi{\'c}}}]{Cavalcanti2017a}%
  \BibitemOpen
  \bibfield  {author} {\bibinfo {author} {\bibfnamefont {D.}~\bibnamefont
  {Cavalcanti}}, \bibinfo {author} {\bibfnamefont {P.}~\bibnamefont
  {Skrzypczyk}}, \ and\ \bibinfo {author} {\bibfnamefont {I.}~\bibnamefont {{\v
  S}upi{\'c}}},\ }\href {\doibase 10.1103/PhysRevLett.119.110501} {\bibfield
  {journal} {\bibinfo  {journal} {Phys. Rev. Lett.}\ }\textbf {\bibinfo
  {volume} {119}},\ \bibinfo {pages} {110501} (\bibinfo {year}
  {2017})}\BibitemShut {NoStop}%
\bibitem [{\citenamefont {Hoban}\ and\ \citenamefont
  {Sainz}(2018)}]{Hoban2018}%
  \BibitemOpen
  \bibfield  {author} {\bibinfo {author} {\bibfnamefont {M.~J.}\ \bibnamefont
  {Hoban}}\ and\ \bibinfo {author} {\bibfnamefont {A.~B.}\ \bibnamefont
  {Sainz}},\ }\href {\doibase 10.1088/1367-2630/aabea8} {\bibfield  {journal}
  {\bibinfo  {journal} {New J. Phys.}\ }\textbf {\bibinfo {volume} {20}},\
  \bibinfo {pages} {053048} (\bibinfo {year} {2018})}\BibitemShut {NoStop}%
\bibitem [{\citenamefont {Buscemi}(2012)}]{Buscemi2012}%
  \BibitemOpen
  \bibfield  {author} {\bibinfo {author} {\bibfnamefont {F.}~\bibnamefont
  {Buscemi}},\ }\href {\doibase 10.1103/PhysRevLett.108.200401} {\bibfield
  {journal} {\bibinfo  {journal} {Phys. Rev. Lett.}\ }\textbf {\bibinfo
  {volume} {108}},\ \bibinfo {pages} {200401} (\bibinfo {year}
  {2012})}\BibitemShut {NoStop}%
\bibitem [{\citenamefont {Cavalcanti}\ \emph {et~al.}(2013)\citenamefont
  {Cavalcanti}, \citenamefont {Hall},\ and\ \citenamefont
  {Wiseman}}]{Cavalcanti2013a}%
  \BibitemOpen
  \bibfield  {author} {\bibinfo {author} {\bibfnamefont {E.~G.}\ \bibnamefont
  {Cavalcanti}}, \bibinfo {author} {\bibfnamefont {M.~J.~W.}\ \bibnamefont
  {Hall}}, \ and\ \bibinfo {author} {\bibfnamefont {H.~M.}\ \bibnamefont
  {Wiseman}},\ }\href {\doibase 10.1103/PhysRevA.87.032306} {\bibfield
  {journal} {\bibinfo  {journal} {Phys. Rev. A}\ }\textbf {\bibinfo {volume}
  {87}},\ \bibinfo {pages} {032306} (\bibinfo {year} {2013})}\BibitemShut
  {NoStop}%
\bibitem [{\citenamefont {Piani}(2015)}]{Piani2015}%
  \BibitemOpen
  \bibfield  {author} {\bibinfo {author} {\bibfnamefont {M.}~\bibnamefont
  {Piani}},\ }\href {\doibase 10.1364/JOSAB.32.0000A1} {\bibfield  {journal}
  {\bibinfo  {journal} {J. Opt. Soc. Am. B, JOSAB}\ }\textbf {\bibinfo {volume}
  {32}},\ \bibinfo {pages} {A1} (\bibinfo {year} {2015})}\BibitemShut {NoStop}%
\bibitem [{\citenamefont {Sainz}\ \emph {et~al.}(2019)\citenamefont {Sainz},
  \citenamefont {Hoban}, \citenamefont {Skrzypczyk},\ and\ \citenamefont
  {Aolita}}]{Sainz2019}%
  \BibitemOpen
  \bibfield  {author} {\bibinfo {author} {\bibfnamefont {A.~B.}\ \bibnamefont
  {Sainz}}, \bibinfo {author} {\bibfnamefont {M.~J.}\ \bibnamefont {Hoban}},
  \bibinfo {author} {\bibfnamefont {P.}~\bibnamefont {Skrzypczyk}}, \ and\
  \bibinfo {author} {\bibfnamefont {L.}~\bibnamefont {Aolita}},\ }\href@noop {}
  {\  (\bibinfo {year} {2019})},\ \Eprint {http://arxiv.org/abs/1907.03705}
  {arXiv:1907.03705 [quant-ph]} \BibitemShut {NoStop}%
\bibitem [{\citenamefont {Watrous}(2018)}]{watrous_2018}%
  \BibitemOpen
  \bibfield  {author} {\bibinfo {author} {\bibfnamefont {J.}~\bibnamefont
  {Watrous}},\ }\href@noop {} {\emph {\bibinfo {title} {The theory of quantum
  information}}}\ (\bibinfo  {publisher} {Cambridge University Press.},\
  \bibinfo {year} {2018})\BibitemShut {NoStop}%
\bibitem [{\citenamefont {Cavalcanti}\ \emph {et~al.}(2015)\citenamefont
  {Cavalcanti}, \citenamefont {Skrzypczyk}, \citenamefont {Aguilar},
  \citenamefont {Nery}, \citenamefont {Ribeiro},\ and\ \citenamefont
  {Walborn}}]{Cavalcanti2015}%
  \BibitemOpen
  \bibfield  {author} {\bibinfo {author} {\bibfnamefont {D.}~\bibnamefont
  {Cavalcanti}}, \bibinfo {author} {\bibfnamefont {P.}~\bibnamefont
  {Skrzypczyk}}, \bibinfo {author} {\bibfnamefont {G.~H.}\ \bibnamefont
  {Aguilar}}, \bibinfo {author} {\bibfnamefont {R.~V.}\ \bibnamefont {Nery}},
  \bibinfo {author} {\bibfnamefont {P.~H.~S.}\ \bibnamefont {Ribeiro}}, \ and\
  \bibinfo {author} {\bibfnamefont {S.~P.}\ \bibnamefont {Walborn}},\ }\href
  {\doibase 10.1038/ncomms8941} {\bibfield  {journal} {\bibinfo  {journal}
  {Nature Communications}\ }\textbf {\bibinfo {volume} {6}},\ \bibinfo {pages}
  {7941} (\bibinfo {year} {2015})}\BibitemShut {NoStop}%
\bibitem [{\citenamefont {{Xu}}\ \emph {et~al.}()\citenamefont {{Xu}},
  \citenamefont {{Zhang}}, \citenamefont {{Lo}},\ and\ \citenamefont
  {{Pan}}}]{wei2019}%
  \BibitemOpen
  \bibfield  {author} {\bibinfo {author} {\bibfnamefont {F.}~\bibnamefont
  {{Xu}}}, \bibinfo {author} {\bibfnamefont {X.~M.~Q.}\ \bibnamefont
  {{Zhang}}}, \bibinfo {author} {\bibfnamefont {H.-K.}\ \bibnamefont {{Lo}}}, \
  and\ \bibinfo {author} {\bibfnamefont {J.-W.}\ \bibnamefont {{Pan}}},\
  }\href@noop {} {\ }\Eprint {http://arxiv.org/abs/1903.09051}
  {arXiv:1903.09051 [quant-ph]} \BibitemShut {NoStop}%
\bibitem [{Note1()}]{Note1}%
  \BibitemOpen
  \bibinfo {note} {Ref.~\cite {Hoban2018} also recognized the value of treating
  various types of resources as channels, but did not study the full range of
  possibilities, as we have done here. The chief differences between our work
  and Ref.~6 are that our work focuses on nonclassicality (while Ref.~\cite
  {Hoban2018} focuses on postquantumness), and that we leverage the resource
  theoretic framework to unify the distinct types of nonclassicality, as
  discussed in the main text.}\BibitemShut {Stop}%
\bibitem [{\citenamefont {Schmid}\ \emph {et~al.}(2020)\citenamefont {Schmid},
  \citenamefont {Rosset},\ and\ \citenamefont {Buscemi}}]{Schmid2019}%
  \BibitemOpen
  \bibfield  {author} {\bibinfo {author} {\bibfnamefont {D.}~\bibnamefont
  {Schmid}}, \bibinfo {author} {\bibfnamefont {D.}~\bibnamefont {Rosset}}, \
  and\ \bibinfo {author} {\bibfnamefont {F.}~\bibnamefont {Buscemi}},\ }\href
  {\doibase 10.22331/q-2020-04-30-262} {\bibfield  {journal} {\bibinfo
  {journal} {{Quantum}}\ }\textbf {\bibinfo {volume} {4}},\ \bibinfo {pages}
  {262} (\bibinfo {year} {2020})}\BibitemShut {NoStop}%
\bibitem [{\citenamefont {Coecke}\ \emph {et~al.}(2016)\citenamefont {Coecke},
  \citenamefont {Fritz},\ and\ \citenamefont {Spekkens}}]{Coecke2016}%
  \BibitemOpen
  \bibfield  {author} {\bibinfo {author} {\bibfnamefont {B.}~\bibnamefont
  {Coecke}}, \bibinfo {author} {\bibfnamefont {T.}~\bibnamefont {Fritz}}, \
  and\ \bibinfo {author} {\bibfnamefont {R.~W.}\ \bibnamefont {Spekkens}},\
  }\href {\doibase 10.1016/j.ic.2016.02.008} {\bibfield  {journal} {\bibinfo
  {journal} {Information and Computation}\ }\bibinfo {series} {Quantum
  {{Physics}} and {{Logic}}},\ \textbf {\bibinfo {volume} {250}},\ \bibinfo
  {pages} {59} (\bibinfo {year} {2016})}\BibitemShut {NoStop}%
\bibitem [{\citenamefont {Wolfe}\ \emph {et~al.}(2020)\citenamefont {Wolfe},
  \citenamefont {Schmid}, \citenamefont {Sainz}, \citenamefont {Kunjwal},\ and\
  \citenamefont {Spekkens}}]{wolfe2019quantifying}%
  \BibitemOpen
  \bibfield  {author} {\bibinfo {author} {\bibfnamefont {E.}~\bibnamefont
  {Wolfe}}, \bibinfo {author} {\bibfnamefont {D.}~\bibnamefont {Schmid}},
  \bibinfo {author} {\bibfnamefont {A.~B.}\ \bibnamefont {Sainz}}, \bibinfo
  {author} {\bibfnamefont {R.}~\bibnamefont {Kunjwal}}, \ and\ \bibinfo
  {author} {\bibfnamefont {R.~W.}\ \bibnamefont {Spekkens}},\ }\href {\doibase
  10.22331/q-2020-06-08-280} {\bibfield  {journal} {\bibinfo  {journal}
  {Quantum}\ }\textbf {\bibinfo {volume} {4}},\ \bibinfo {pages} {280}
  (\bibinfo {year} {2020})}\BibitemShut {NoStop}%
\bibitem [{\citenamefont {Schmid}\ \emph
  {et~al.}(2019{\natexlab{a}})\citenamefont {Schmid}, \citenamefont {Fraser},
  \citenamefont {Kunjwal}, \citenamefont {Sainz}, \citenamefont {Wolfe},\ and\
  \citenamefont {Spekkens}}]{LOSRvsLOCCentang}%
  \BibitemOpen
  \bibfield  {author} {\bibinfo {author} {\bibfnamefont {D.}~\bibnamefont
  {Schmid}}, \bibinfo {author} {\bibfnamefont {T.~C.}\ \bibnamefont {Fraser}},
  \bibinfo {author} {\bibfnamefont {R.}~\bibnamefont {Kunjwal}}, \bibinfo
  {author} {\bibfnamefont {A.~B.}\ \bibnamefont {Sainz}}, \bibinfo {author}
  {\bibfnamefont {E.}~\bibnamefont {Wolfe}}, \ and\ \bibinfo {author}
  {\bibfnamefont {R.~W.}\ \bibnamefont {Spekkens}},\ }\href@noop {} {}
  (\bibinfo {year} {2019}{\natexlab{a}}),\ \Eprint
  {http://arxiv.org/abs/2004.09194} {arXiv:2004.09194 [quant-ph]} \BibitemShut
  {NoStop}%
\bibitem [{\citenamefont {Nielsen}\ and\ \citenamefont
  {Chuang}(2010)}]{NielsenAndChuang}%
  \BibitemOpen
  \bibfield  {author} {\bibinfo {author} {\bibfnamefont {M.~A.}\ \bibnamefont
  {Nielsen}}\ and\ \bibinfo {author} {\bibfnamefont {I.~L.}\ \bibnamefont
  {Chuang}},\ }\href {\doibase 10.1017/CBO9780511976667} {\emph {\bibinfo
  {title} {{Quantum Computation and Quantum Information}}}}\ (\bibinfo
  {publisher} {Cambridge University Press},\ \bibinfo {year}
  {2010})\BibitemShut {NoStop}%
\bibitem [{\citenamefont {Schmid}\ \emph
  {et~al.}(2019{\natexlab{b}})\citenamefont {Schmid}, \citenamefont {Ried},\
  and\ \citenamefont {Spekkens}}]{Schmidcausal}%
  \BibitemOpen
  \bibfield  {author} {\bibinfo {author} {\bibfnamefont {D.}~\bibnamefont
  {Schmid}}, \bibinfo {author} {\bibfnamefont {K.}~\bibnamefont {Ried}}, \ and\
  \bibinfo {author} {\bibfnamefont {R.~W.}\ \bibnamefont {Spekkens}},\ }\href
  {\doibase 10.1103/PhysRevA.100.022112} {\bibfield  {journal} {\bibinfo
  {journal} {Phys. Rev. A}\ }\textbf {\bibinfo {volume} {100}},\ \bibinfo
  {pages} {022112} (\bibinfo {year} {2019}{\natexlab{b}})}\BibitemShut
  {NoStop}%
\bibitem [{\citenamefont {Choi}(1975)}]{Choi1975}%
  \BibitemOpen
  \bibfield  {author} {\bibinfo {author} {\bibfnamefont {M.-D.}\ \bibnamefont
  {Choi}},\ }\href {\doibase 10.1016/0024-3795(75)90075-0} {\bibfield
  {journal} {\bibinfo  {journal} {Linear Algebra and its Applications}\
  }\textbf {\bibinfo {volume} {10}},\ \bibinfo {pages} {285} (\bibinfo {year}
  {1975})}\BibitemShut {NoStop}%
\bibitem [{\citenamefont {Jamio{\l}kowski}(1972)}]{Jamiolkowski1972}%
  \BibitemOpen
  \bibfield  {author} {\bibinfo {author} {\bibfnamefont {A.}~\bibnamefont
  {Jamio{\l}kowski}},\ }\href {\doibase 10.1016/0034-4877(72)90011-0}
  {\bibfield  {journal} {\bibinfo  {journal} {Reports on Mathematical Physics}\
  }\textbf {\bibinfo {volume} {3}},\ \bibinfo {pages} {275} (\bibinfo {year}
  {1972})}\BibitemShut {NoStop}%
\bibitem [{Sup()}]{SuppMat}%
  \BibitemOpen
  \href@noop {} {}\bibinfo {note} {See Supplemental Material at [URL will be
  inserted by publisher] for technical details regarding our techniques. It
  includes Refs. \cite{ApS2015}-\cite{Vidal2002}}\BibitemShut {NoStop}%
\bibitem [{\citenamefont {Verbanis}\ \emph {et~al.}(2016)\citenamefont
  {Verbanis}, \citenamefont {Martin}, \citenamefont {Rosset}, \citenamefont
  {Lim}, \citenamefont {Thew},\ and\ \citenamefont {Zbinden}}]{Verbanis2016}%
  \BibitemOpen
  \bibfield  {author} {\bibinfo {author} {\bibfnamefont {E.}~\bibnamefont
  {Verbanis}}, \bibinfo {author} {\bibfnamefont {A.}~\bibnamefont {Martin}},
  \bibinfo {author} {\bibfnamefont {D.}~\bibnamefont {Rosset}}, \bibinfo
  {author} {\bibfnamefont {C.~C.~W.}\ \bibnamefont {Lim}}, \bibinfo {author}
  {\bibfnamefont {R.~T.}\ \bibnamefont {Thew}}, \ and\ \bibinfo {author}
  {\bibfnamefont {H.}~\bibnamefont {Zbinden}},\ }\href {\doibase
  10.1103/PhysRevLett.116.190501} {\bibfield  {journal} {\bibinfo  {journal}
  {Phys. Rev. Lett.}\ }\textbf {\bibinfo {volume} {116}},\ \bibinfo {pages}
  {190501} (\bibinfo {year} {2016})}\BibitemShut {NoStop}%
\bibitem [{\citenamefont {Rosset}\ \emph {et~al.}(2018)\citenamefont {Rosset},
  \citenamefont {Martin}, \citenamefont {Verbanis}, \citenamefont {Lim},\ and\
  \citenamefont {Thew}}]{Rosset2018a}%
  \BibitemOpen
  \bibfield  {author} {\bibinfo {author} {\bibfnamefont {D.}~\bibnamefont
  {Rosset}}, \bibinfo {author} {\bibfnamefont {A.}~\bibnamefont {Martin}},
  \bibinfo {author} {\bibfnamefont {E.}~\bibnamefont {Verbanis}}, \bibinfo
  {author} {\bibfnamefont {C.~C.~W.}\ \bibnamefont {Lim}}, \ and\ \bibinfo
  {author} {\bibfnamefont {R.}~\bibnamefont {Thew}},\ }\href {\doibase
  10.1103/PhysRevA.98.052332} {\bibfield  {journal} {\bibinfo  {journal} {Phys.
  Rev. A}\ }\textbf {\bibinfo {volume} {98}},\ \bibinfo {pages} {052332}
  (\bibinfo {year} {2018})}\BibitemShut {NoStop}%
\bibitem [{\citenamefont {Tsirel'son}(1987)}]{Tsirelson1987}%
  \BibitemOpen
  \bibfield  {author} {\bibinfo {author} {\bibfnamefont {B.~S.}\ \bibnamefont
  {Tsirel'son}},\ }\href {\doibase 10.1007/BF01663472} {\bibfield  {journal}
  {\bibinfo  {journal} {Journal of Soviet Mathematics}\ }\textbf {\bibinfo
  {volume} {36}},\ \bibinfo {pages} {557} (\bibinfo {year} {1987})}\BibitemShut
  {NoStop}%
\bibitem [{\citenamefont {Chiribella}\ \emph {et~al.}(2008)\citenamefont
  {Chiribella}, \citenamefont {D'Ariano},\ and\ \citenamefont
  {Perinotti}}]{Chiribella2008}%
  \BibitemOpen
  \bibfield  {author} {\bibinfo {author} {\bibfnamefont {G.}~\bibnamefont
  {Chiribella}}, \bibinfo {author} {\bibfnamefont {G.~M.}\ \bibnamefont
  {D'Ariano}}, \ and\ \bibinfo {author} {\bibfnamefont {P.}~\bibnamefont
  {Perinotti}},\ }\href {\doibase 10.1209/0295-5075/83/30004} {\bibfield
  {journal} {\bibinfo  {journal} {EPL (Europhysics Letters)}\ }\textbf
  {\bibinfo {volume} {83}},\ \bibinfo {pages} {30004} (\bibinfo {year}
  {2008})}\BibitemShut {NoStop}%
\bibitem [{\citenamefont {Vidal}\ and\ \citenamefont
  {Tarrach}(1999)}]{Vidal1999}%
  \BibitemOpen
  \bibfield  {author} {\bibinfo {author} {\bibfnamefont {G.}~\bibnamefont
  {Vidal}}\ and\ \bibinfo {author} {\bibfnamefont {R.}~\bibnamefont
  {Tarrach}},\ }\href {\doibase 10.1103/PhysRevA.59.141} {\bibfield  {journal}
  {\bibinfo  {journal} {Phys. Rev. A}\ }\textbf {\bibinfo {volume} {59}},\
  \bibinfo {pages} {141} (\bibinfo {year} {1999})}\BibitemShut {NoStop}%
\bibitem [{\citenamefont {Geller}\ and\ \citenamefont
  {Piani}(2014)}]{Geller2014}%
  \BibitemOpen
  \bibfield  {author} {\bibinfo {author} {\bibfnamefont {J.}~\bibnamefont
  {Geller}}\ and\ \bibinfo {author} {\bibfnamefont {M.}~\bibnamefont {Piani}},\
  }\href {\doibase 10.1088/1751-8113/47/42/424030} {\bibfield  {journal}
  {\bibinfo  {journal} {Journal of Physics A: Mathematical and Theoretical}\
  }\textbf {\bibinfo {volume} {47}},\ \bibinfo {pages} {424030} (\bibinfo
  {year} {2014})}\BibitemShut {NoStop}%
\bibitem [{\citenamefont {Cavalcanti}\ and\ \citenamefont
  {Skrzypczyk}(2016)}]{Cavalcanti2016a}%
  \BibitemOpen
  \bibfield  {author} {\bibinfo {author} {\bibfnamefont {D.}~\bibnamefont
  {Cavalcanti}}\ and\ \bibinfo {author} {\bibfnamefont {P.}~\bibnamefont
  {Skrzypczyk}},\ }\href {\doibase 10.1103/PhysRevA.93.052112} {\bibfield
  {journal} {\bibinfo  {journal} {Physical Review A}\ }\textbf {\bibinfo
  {volume} {93}},\ \bibinfo {pages} {052112} (\bibinfo {year}
  {2016})}\BibitemShut {NoStop}%
\bibitem [{\citenamefont {Sainz}\ \emph {et~al.}(2016)\citenamefont {Sainz},
  \citenamefont {Aolita}, \citenamefont {Brunner}, \citenamefont {Gallego},\
  and\ \citenamefont {Skrzypczyk}}]{Sainz2016}%
  \BibitemOpen
  \bibfield  {author} {\bibinfo {author} {\bibfnamefont {A.~B.}\ \bibnamefont
  {Sainz}}, \bibinfo {author} {\bibfnamefont {L.}~\bibnamefont {Aolita}},
  \bibinfo {author} {\bibfnamefont {N.}~\bibnamefont {Brunner}}, \bibinfo
  {author} {\bibfnamefont {R.}~\bibnamefont {Gallego}}, \ and\ \bibinfo
  {author} {\bibfnamefont {P.}~\bibnamefont {Skrzypczyk}},\ }\href {\doibase
  10.1103/PhysRevA.94.012308} {\bibfield  {journal} {\bibinfo  {journal}
  {Physical Review A}\ }\textbf {\bibinfo {volume} {94}},\ \bibinfo {pages}
  {012308} (\bibinfo {year} {2016})}\BibitemShut {NoStop}%
\bibitem [{\citenamefont {de~Vicente}(2014)}]{Vicente2014}%
  \BibitemOpen
  \bibfield  {author} {\bibinfo {author} {\bibfnamefont {J.~I.}\ \bibnamefont
  {de~Vicente}},\ }\href {\doibase 10.1088/1751-8113/47/42/424017} {\bibfield
  {journal} {\bibinfo  {journal} {Journal of Physics A: Mathematical and
  Theoretical}\ }\textbf {\bibinfo {volume} {47}},\ \bibinfo {pages} {424017}
  (\bibinfo {year} {2014})}\BibitemShut {NoStop}%
\bibitem [{\citenamefont {Doherty}\ \emph {et~al.}(2004)\citenamefont
  {Doherty}, \citenamefont {Parrilo},\ and\ \citenamefont
  {Spedalieri}}]{Doherty2004}%
  \BibitemOpen
  \bibfield  {author} {\bibinfo {author} {\bibfnamefont {A.~C.}\ \bibnamefont
  {Doherty}}, \bibinfo {author} {\bibfnamefont {P.~A.}\ \bibnamefont
  {Parrilo}}, \ and\ \bibinfo {author} {\bibfnamefont {F.~M.}\ \bibnamefont
  {Spedalieri}},\ }\href {\doibase 10.1103/PhysRevA.69.022308} {\bibfield
  {journal} {\bibinfo  {journal} {Phys. Rev. A}\ }\textbf {\bibinfo {volume}
  {69}},\ \bibinfo {pages} {022308} (\bibinfo {year} {2004})}\BibitemShut
  {NoStop}%
\bibitem [{\citenamefont {Berta}\ \emph {et~al.}(2018)\citenamefont {Berta},
  \citenamefont {Borderi}, \citenamefont {Fawzi},\ and\ \citenamefont
  {Scholz}}]{Berta2018}%
  \BibitemOpen
  \bibfield  {author} {\bibinfo {author} {\bibfnamefont {M.}~\bibnamefont
  {Berta}}, \bibinfo {author} {\bibfnamefont {F.}~\bibnamefont {Borderi}},
  \bibinfo {author} {\bibfnamefont {O.}~\bibnamefont {Fawzi}}, \ and\ \bibinfo
  {author} {\bibfnamefont {V.}~\bibnamefont {Scholz}},\ }\href@noop {} {\
  (\bibinfo {year} {2018})},\ \Eprint {http://arxiv.org/abs/1810.12197}
  {arXiv:1810.12197 [quant-ph]} \BibitemShut {NoStop}%
\bibitem [{\citenamefont {Bowles}\ \emph {et~al.}(2015)\citenamefont {Bowles},
  \citenamefont {Hirsch}, \citenamefont {Quintino},\ and\ \citenamefont
  {Brunner}}]{Bowles2015}%
  \BibitemOpen
  \bibfield  {author} {\bibinfo {author} {\bibfnamefont {J.}~\bibnamefont
  {Bowles}}, \bibinfo {author} {\bibfnamefont {F.}~\bibnamefont {Hirsch}},
  \bibinfo {author} {\bibfnamefont {M.~T.}\ \bibnamefont {Quintino}}, \ and\
  \bibinfo {author} {\bibfnamefont {N.}~\bibnamefont {Brunner}},\ }\href
  {\doibase 10.1103/PhysRevLett.114.120401} {\bibfield  {journal} {\bibinfo
  {journal} {Phys. Rev. Lett.}\ }\textbf {\bibinfo {volume} {114}},\ \bibinfo
  {pages} {120401} (\bibinfo {year} {2015})}\BibitemShut {NoStop}%
\bibitem [{\citenamefont {Bancal}\ \emph {et~al.}(2011)\citenamefont {Bancal},
  \citenamefont {Gisin}, \citenamefont {Liang},\ and\ \citenamefont
  {Pironio}}]{Bancal2011}%
  \BibitemOpen
  \bibfield  {author} {\bibinfo {author} {\bibfnamefont {J.-D.}\ \bibnamefont
  {Bancal}}, \bibinfo {author} {\bibfnamefont {N.}~\bibnamefont {Gisin}},
  \bibinfo {author} {\bibfnamefont {Y.-C.}\ \bibnamefont {Liang}}, \ and\
  \bibinfo {author} {\bibfnamefont {S.}~\bibnamefont {Pironio}},\ }\href
  {\doibase 10.1103/PhysRevLett.106.250404} {\bibfield  {journal} {\bibinfo
  {journal} {Phys. Rev. Lett.}\ }\textbf {\bibinfo {volume} {106}},\ \bibinfo
  {pages} {250404} (\bibinfo {year} {2011})}\BibitemShut {NoStop}%
\bibitem [{\citenamefont {Gallego}\ \emph {et~al.}(2012)\citenamefont
  {Gallego}, \citenamefont {W\"urflinger}, \citenamefont {Ac\'{\i}n},\ and\
  \citenamefont {Navascu\'es}}]{Gallego2012}%
  \BibitemOpen
  \bibfield  {author} {\bibinfo {author} {\bibfnamefont {R.}~\bibnamefont
  {Gallego}}, \bibinfo {author} {\bibfnamefont {L.~E.}\ \bibnamefont
  {W\"urflinger}}, \bibinfo {author} {\bibfnamefont {A.}~\bibnamefont
  {Ac\'{\i}n}}, \ and\ \bibinfo {author} {\bibfnamefont {M.}~\bibnamefont
  {Navascu\'es}},\ }\href {\doibase 10.1103/PhysRevLett.109.070401} {\bibfield
  {journal} {\bibinfo  {journal} {Phys. Rev. Lett.}\ }\textbf {\bibinfo
  {volume} {109}},\ \bibinfo {pages} {070401} (\bibinfo {year}
  {2012})}\BibitemShut {NoStop}%
\bibitem [{\citenamefont {Bancal}\ \emph {et~al.}(2013)\citenamefont {Bancal},
  \citenamefont {Barrett}, \citenamefont {Gisin},\ and\ \citenamefont
  {Pironio}}]{Bancal2013}%
  \BibitemOpen
  \bibfield  {author} {\bibinfo {author} {\bibfnamefont {J.-D.}\ \bibnamefont
  {Bancal}}, \bibinfo {author} {\bibfnamefont {J.}~\bibnamefont {Barrett}},
  \bibinfo {author} {\bibfnamefont {N.}~\bibnamefont {Gisin}}, \ and\ \bibinfo
  {author} {\bibfnamefont {S.}~\bibnamefont {Pironio}},\ }\href {\doibase
  10.1103/PhysRevA.88.014102} {\bibfield  {journal} {\bibinfo  {journal} {Phys.
  Rev. A}\ }\textbf {\bibinfo {volume} {88}},\ \bibinfo {pages} {014102}
  (\bibinfo {year} {2013})}\BibitemShut {NoStop}%
\bibitem [{\citenamefont {D{\"u}r}\ \emph {et~al.}(2000)\citenamefont
  {D{\"u}r}, \citenamefont {Vidal},\ and\ \citenamefont {Cirac}}]{Dur2000}%
  \BibitemOpen
  \bibfield  {author} {\bibinfo {author} {\bibfnamefont {W.}~\bibnamefont
  {D{\"u}r}}, \bibinfo {author} {\bibfnamefont {G.}~\bibnamefont {Vidal}}, \
  and\ \bibinfo {author} {\bibfnamefont {J.~I.}\ \bibnamefont {Cirac}},\ }\href
  {\doibase 10.1103/PhysRevA.62.062314} {\bibfield  {journal} {\bibinfo
  {journal} {Physical Review A}\ }\textbf {\bibinfo {volume} {62}},\ \bibinfo
  {pages} {062314} (\bibinfo {year} {2000})}\BibitemShut {NoStop}%
\bibitem [{\citenamefont {{Schmid}}\ \emph {et~al.}(2020)\citenamefont
  {{Schmid}}, \citenamefont {{Du}}, \citenamefont {{Mudassar}}, \citenamefont
  {{Coulter-de Wit}}, \citenamefont {{Rosset}},\ and\ \citenamefont
  {{Hoban}}}]{postquantumschmid}%
  \BibitemOpen
  \bibfield  {author} {\bibinfo {author} {\bibfnamefont {D.}~\bibnamefont
  {{Schmid}}}, \bibinfo {author} {\bibfnamefont {H.}~\bibnamefont {{Du}}},
  \bibinfo {author} {\bibfnamefont {M.}~\bibnamefont {{Mudassar}}}, \bibinfo
  {author} {\bibfnamefont {G.}~\bibnamefont {{Coulter-de Wit}}}, \bibinfo
  {author} {\bibfnamefont {D.}~\bibnamefont {{Rosset}}}, \ and\ \bibinfo
  {author} {\bibfnamefont {M.~J.}\ \bibnamefont {{Hoban}}},\ }\href@noop {} {}
  (\bibinfo {year} {2020}),\ \Eprint {http://arxiv.org/abs/2004.06133}
  {arXiv:2004.06133 [quant-ph]} \BibitemShut {NoStop}%
\bibitem [{\citenamefont {Gonda}\ and\ \citenamefont
  {Spekkens}(2019)}]{Gonda2019}%
  \BibitemOpen
  \bibfield  {author} {\bibinfo {author} {\bibfnamefont {T.}~\bibnamefont
  {Gonda}}\ and\ \bibinfo {author} {\bibfnamefont {R.~W.}\ \bibnamefont
  {Spekkens}},\ }\href@noop {} {\  (\bibinfo {year} {2019})},\ \Eprint
  {http://arxiv.org/abs/1912.07085} {arXiv:1912.07085 [math-ph, quant-ph]}
  \BibitemShut {NoStop}%
\bibitem [{\citenamefont {ApS}(2015)}]{ApS2015}%
  \BibitemOpen
  \bibfield  {author} {\bibinfo {author} {\bibfnamefont {M.}~\bibnamefont
  {ApS}},\ }\href@noop {} {\emph {\bibinfo {title} {The {{MOSEK}} Optimization
  Toolbox for {{MATLAB}} Manual. {{Version}} 7.1 ({{Revision}} 28).}}}\
  (\bibinfo {year} {2015})\BibitemShut {NoStop}%
\bibitem [{\citenamefont {Beckman}\ \emph {et~al.}(2001)\citenamefont
  {Beckman}, \citenamefont {Gottesman}, \citenamefont {Nielsen},\ and\
  \citenamefont {Preskill}}]{Beckman2001}%
  \BibitemOpen
  \bibfield  {author} {\bibinfo {author} {\bibfnamefont {D.}~\bibnamefont
  {Beckman}}, \bibinfo {author} {\bibfnamefont {D.}~\bibnamefont {Gottesman}},
  \bibinfo {author} {\bibfnamefont {M.~A.}\ \bibnamefont {Nielsen}}, \ and\
  \bibinfo {author} {\bibfnamefont {J.}~\bibnamefont {Preskill}},\ }\href
  {\doibase 10.1103/PhysRevA.64.052309} {\bibfield  {journal} {\bibinfo
  {journal} {Physical Review A}\ }\textbf {\bibinfo {volume} {64}},\ \bibinfo
  {pages} {052309} (\bibinfo {year} {2001})}\BibitemShut {NoStop}%
\bibitem [{\citenamefont {Boyd}\ and\ \citenamefont
  {Vandenberghe}(2004)}]{Boyd2004}%
  \BibitemOpen
  \bibfield  {author} {\bibinfo {author} {\bibfnamefont {S.}~\bibnamefont
  {Boyd}}\ and\ \bibinfo {author} {\bibfnamefont {L.}~\bibnamefont
  {Vandenberghe}},\ }\href@noop {} {\emph {\bibinfo {title} {{Convex
  Optimization}}}},\ \bibinfo {edition} {new.}\ ed.\ (\bibinfo  {publisher}
  {{Cambridge University Press}},\ \bibinfo {address} {{Cambridge, UK}},\
  \bibinfo {year} {2004})\BibitemShut {NoStop}%
\bibitem [{\citenamefont {D'Ariano}\ \emph {et~al.}(2011)\citenamefont
  {D'Ariano}, \citenamefont {Facchini},\ and\ \citenamefont
  {Perinotti}}]{DAriano2011}%
  \BibitemOpen
  \bibfield  {author} {\bibinfo {author} {\bibfnamefont {G.~M.}\ \bibnamefont
  {D'Ariano}}, \bibinfo {author} {\bibfnamefont {S.}~\bibnamefont {Facchini}},
  \ and\ \bibinfo {author} {\bibfnamefont {P.}~\bibnamefont {Perinotti}},\
  }\href {\doibase 10.1103/PhysRevLett.106.010501} {\bibfield  {journal}
  {\bibinfo  {journal} {Physical Review Letters}\ }\textbf {\bibinfo {volume}
  {106}},\ \bibinfo {pages} {010501} (\bibinfo {year} {2011})}\BibitemShut
  {NoStop}%
\bibitem [{\citenamefont {van Dam}\ \emph {et~al.}(2005)\citenamefont {van
  Dam}, \citenamefont {Gill},\ and\ \citenamefont {Grunwald}}]{Dam2005}%
  \BibitemOpen
  \bibfield  {author} {\bibinfo {author} {\bibfnamefont {W.}~\bibnamefont {van
  Dam}}, \bibinfo {author} {\bibfnamefont {R.~D.}\ \bibnamefont {Gill}}, \ and\
  \bibinfo {author} {\bibfnamefont {P.~D.}\ \bibnamefont {Grunwald}},\ }\href
  {\doibase 10.1109/TIT.2005.851738} {\bibfield  {journal} {\bibinfo  {journal}
  {IEEE Transactions on Information Theory}\ }\textbf {\bibinfo {volume}
  {51}},\ \bibinfo {pages} {2812} (\bibinfo {year} {2005})}\BibitemShut
  {NoStop}%
\bibitem [{\citenamefont {Gallego}\ and\ \citenamefont
  {Aolita}(2015)}]{Gallego2015}%
  \BibitemOpen
  \bibfield  {author} {\bibinfo {author} {\bibfnamefont {R.}~\bibnamefont
  {Gallego}}\ and\ \bibinfo {author} {\bibfnamefont {L.}~\bibnamefont
  {Aolita}},\ }\href {\doibase 10.1103/PhysRevX.5.041008} {\bibfield  {journal}
  {\bibinfo  {journal} {Physical Review X}\ }\textbf {\bibinfo {volume} {5}},\
  \bibinfo {pages} {041008} (\bibinfo {year} {2015})}\BibitemShut {NoStop}%
\bibitem [{\citenamefont {Garg}\ and\ \citenamefont {Mermin}(1984)}]{Garg1984}%
  \BibitemOpen
  \bibfield  {author} {\bibinfo {author} {\bibfnamefont {A.}~\bibnamefont
  {Garg}}\ and\ \bibinfo {author} {\bibfnamefont {N.~D.}\ \bibnamefont
  {Mermin}},\ }\href {\doibase 10.1007/BF00741645} {\bibfield  {journal}
  {\bibinfo  {journal} {Foundations of Physics}\ }\textbf {\bibinfo {volume}
  {14}},\ \bibinfo {pages} {1} (\bibinfo {year} {1984})}\BibitemShut {NoStop}%
\bibitem [{\citenamefont {Helton}\ and\ \citenamefont
  {Nie}(2009)}]{Helton2009}%
  \BibitemOpen
  \bibfield  {author} {\bibinfo {author} {\bibfnamefont {J.}~\bibnamefont
  {Helton}}\ and\ \bibinfo {author} {\bibfnamefont {J.}~\bibnamefont {Nie}},\
  }\href {\doibase 10.1137/07070526X} {\bibfield  {journal} {\bibinfo
  {journal} {SIAM Journal on Optimization}\ }\textbf {\bibinfo {volume} {20}},\
  \bibinfo {pages} {759} (\bibinfo {year} {2009})}\BibitemShut {NoStop}%
\bibitem [{\citenamefont {Horodecki}\ \emph {et~al.}(1996)\citenamefont
  {Horodecki}, \citenamefont {Horodecki},\ and\ \citenamefont
  {Horodecki}}]{Horodecki1996}%
  \BibitemOpen
  \bibfield  {author} {\bibinfo {author} {\bibfnamefont {M.}~\bibnamefont
  {Horodecki}}, \bibinfo {author} {\bibfnamefont {P.}~\bibnamefont
  {Horodecki}}, \ and\ \bibinfo {author} {\bibfnamefont {R.}~\bibnamefont
  {Horodecki}},\ }\href {\doibase 10.1016/S0375-9601(96)00706-2} {\bibfield
  {journal} {\bibinfo  {journal} {Physics Letters A}\ }\textbf {\bibinfo
  {volume} {223}},\ \bibinfo {pages} {1} (\bibinfo {year} {1996})}\BibitemShut
  {NoStop}%
\bibitem [{\citenamefont {Johnston}(2016)}]{Johnston2016}%
  \BibitemOpen
  \bibfield  {author} {\bibinfo {author} {\bibfnamefont {N.}~\bibnamefont
  {Johnston}},\ }\href {\doibase 10.5281/zenodo.44637} {\  (\bibinfo {year}
  {2016}),\ 10.5281/zenodo.44637}\BibitemShut {NoStop}%
\bibitem [{\citenamefont {Lewenstein}\ and\ \citenamefont
  {Sanpera}(1998)}]{Lewenstein1998}%
  \BibitemOpen
  \bibfield  {author} {\bibinfo {author} {\bibfnamefont {M.}~\bibnamefont
  {Lewenstein}}\ and\ \bibinfo {author} {\bibfnamefont {A.}~\bibnamefont
  {Sanpera}},\ }\href {\doibase 10.1103/PhysRevLett.80.2261} {\bibfield
  {journal} {\bibinfo  {journal} {Physical Review Letters}\ }\textbf {\bibinfo
  {volume} {80}},\ \bibinfo {pages} {2261} (\bibinfo {year}
  {1998})}\BibitemShut {NoStop}%
\bibitem [{\citenamefont {Lofberg}(2004)}]{Lofberg2004}%
  \BibitemOpen
  \bibfield  {author} {\bibinfo {author} {\bibfnamefont {J.}~\bibnamefont
  {Lofberg}},\ }in\ \href {\doibase 10.1109/CACSD.2004.1393890} {\emph
  {\bibinfo {booktitle} {2004 {{IEEE International Conference}} on {{Robotics}}
  and {{Automation}} ({{IEEE Cat}}. {{No}}.{{04CH37508}})}}}\ (\bibinfo {year}
  {2004})\ pp.\ \bibinfo {pages} {284--289}\BibitemShut {NoStop}%
\bibitem [{\citenamefont {MATLAB}()}]{MATLAB}%
  \BibitemOpen
  \bibfield  {author} {\bibinfo {author} {\bibnamefont {MATLAB}},\ }\href@noop
  {} {}\ (\bibinfo  {publisher} {The MathWorks Inc.},\ \bibinfo {address}
  {Natick, Massachusetts})\BibitemShut {NoStop}%
\bibitem [{\citenamefont {Navascu{\'e}s}\ \emph {et~al.}(2009)\citenamefont
  {Navascu{\'e}s}, \citenamefont {Owari},\ and\ \citenamefont
  {Plenio}}]{Navascues2009}%
  \BibitemOpen
  \bibfield  {author} {\bibinfo {author} {\bibfnamefont {M.}~\bibnamefont
  {Navascu{\'e}s}}, \bibinfo {author} {\bibfnamefont {M.}~\bibnamefont
  {Owari}}, \ and\ \bibinfo {author} {\bibfnamefont {M.~B.}\ \bibnamefont
  {Plenio}},\ }\href {\doibase 10.1103/PhysRevLett.103.160404} {\bibfield
  {journal} {\bibinfo  {journal} {Physical Review Letters}\ }\textbf {\bibinfo
  {volume} {103}},\ \bibinfo {pages} {160404} (\bibinfo {year}
  {2009})}\BibitemShut {NoStop}%
\bibitem [{\citenamefont {Eaton}\ \emph {et~al.}(2020)\citenamefont {Eaton},
  \citenamefont {Bateman}, \citenamefont {Hauberg},\ and\ \citenamefont
  {Wehbring}}]{Octave520}%
  \BibitemOpen
  \bibfield  {author} {\bibinfo {author} {\bibfnamefont {J.~W.}\ \bibnamefont
  {Eaton}}, \bibinfo {author} {\bibfnamefont {D.}~\bibnamefont {Bateman}},
  \bibinfo {author} {\bibfnamefont {S.}~\bibnamefont {Hauberg}}, \ and\
  \bibinfo {author} {\bibfnamefont {R.}~\bibnamefont {Wehbring}},\ }\href
  {https://www.gnu.org/software/octave/doc/v5.2.0/} {\emph {\bibinfo {title}
  {{GNU Octave} version 5.2.0 manual: a high-level interactive language for
  numerical computations}}} (\bibinfo {year} {2020})\BibitemShut {NoStop}%
\bibitem [{\citenamefont {Ostermann}\ and\ \citenamefont
  {Wanner}(2012)}]{Ostermann2012}%
  \BibitemOpen
  \bibfield  {author} {\bibinfo {author} {\bibfnamefont {A.}~\bibnamefont
  {Ostermann}}\ and\ \bibinfo {author} {\bibfnamefont {G.}~\bibnamefont
  {Wanner}},\ }in\ \href {\doibase 10.1007/978-3-642-29163-0_1} {\emph
  {\bibinfo {booktitle} {Geometry by {{Its History}}}}},\ \bibinfo {series and
  number} {Undergraduate {{Texts}} in {{Mathematics}}},\ \bibinfo {editor}
  {edited by\ \bibinfo {editor} {\bibfnamefont {A.}~\bibnamefont {Ostermann}}\
  and\ \bibinfo {editor} {\bibfnamefont {G.}~\bibnamefont {Wanner}}}\ (\bibinfo
   {publisher} {{Springer}},\ \bibinfo {address} {{Berlin, Heidelberg}},\
  \bibinfo {year} {2012})\ pp.\ \bibinfo {pages} {3--26}\BibitemShut {NoStop}%
\bibitem [{\citenamefont {Piani}\ and\ \citenamefont
  {Watrous}(2015)}]{Piani2015a}%
  \BibitemOpen
  \bibfield  {author} {\bibinfo {author} {\bibfnamefont {M.}~\bibnamefont
  {Piani}}\ and\ \bibinfo {author} {\bibfnamefont {J.}~\bibnamefont
  {Watrous}},\ }\href {\doibase 10.1103/PhysRevLett.114.060404} {\bibfield
  {journal} {\bibinfo  {journal} {Physical Review Letters}\ }\textbf {\bibinfo
  {volume} {114}},\ \bibinfo {pages} {060404} (\bibinfo {year}
  {2015})}\BibitemShut {NoStop}%
\bibitem [{\citenamefont {Pitowsky}(1991)}]{Pitowsky1991}%
  \BibitemOpen
  \bibfield  {author} {\bibinfo {author} {\bibfnamefont {I.}~\bibnamefont
  {Pitowsky}},\ }\href {\doibase 10.1007/BF01594946} {\bibfield  {journal}
  {\bibinfo  {journal} {Mathematical Programming}\ }\textbf {\bibinfo {volume}
  {50}},\ \bibinfo {pages} {395} (\bibinfo {year} {1991})}\BibitemShut
  {NoStop}%
\bibitem [{\citenamefont {Portmann}\ \emph {et~al.}(2012)\citenamefont
  {Portmann}, \citenamefont {Branciard},\ and\ \citenamefont
  {Gisin}}]{Portmann2012}%
  \BibitemOpen
  \bibfield  {author} {\bibinfo {author} {\bibfnamefont {S.}~\bibnamefont
  {Portmann}}, \bibinfo {author} {\bibfnamefont {C.}~\bibnamefont {Branciard}},
  \ and\ \bibinfo {author} {\bibfnamefont {N.}~\bibnamefont {Gisin}},\ }\href
  {\doibase 10.1103/PhysRevA.86.012104} {\bibfield  {journal} {\bibinfo
  {journal} {Physical Review A}\ }\textbf {\bibinfo {volume} {86}},\ \bibinfo
  {pages} {012104} (\bibinfo {year} {2012})}\BibitemShut {NoStop}%
\bibitem [{\citenamefont {Skrzypczyk}\ \emph {et~al.}(2014)\citenamefont
  {Skrzypczyk}, \citenamefont {Navascu{\'e}s},\ and\ \citenamefont
  {Cavalcanti}}]{Skrzypczyk2014}%
  \BibitemOpen
  \bibfield  {author} {\bibinfo {author} {\bibfnamefont {P.}~\bibnamefont
  {Skrzypczyk}}, \bibinfo {author} {\bibfnamefont {M.}~\bibnamefont
  {Navascu{\'e}s}}, \ and\ \bibinfo {author} {\bibfnamefont {D.}~\bibnamefont
  {Cavalcanti}},\ }\href {\doibase 10.1103/PhysRevLett.112.180404} {\bibfield
  {journal} {\bibinfo  {journal} {Physical Review Letters}\ }\textbf {\bibinfo
  {volume} {112}},\ \bibinfo {pages} {180404} (\bibinfo {year}
  {2014})}\BibitemShut {NoStop}%
\bibitem [{\citenamefont {Slofstra}(2017)}]{Slofstra2017a}%
  \BibitemOpen
  \bibfield  {author} {\bibinfo {author} {\bibfnamefont {W.}~\bibnamefont
  {Slofstra}},\ }\href@noop {} {\  (\bibinfo {year} {2017})},\ \Eprint
  {http://arxiv.org/abs/1703.08618} {arXiv:1703.08618} \BibitemShut {NoStop}%
\bibitem [{\citenamefont {Steiner}(2003)}]{Steiner2003}%
  \BibitemOpen
  \bibfield  {author} {\bibinfo {author} {\bibfnamefont {M.}~\bibnamefont
  {Steiner}},\ }\href {\doibase 10.1103/PhysRevA.67.054305} {\bibfield
  {journal} {\bibinfo  {journal} {Physical Review A}\ }\textbf {\bibinfo
  {volume} {67}},\ \bibinfo {pages} {054305} (\bibinfo {year}
  {2003})}\BibitemShut {NoStop}%
\bibitem [{\citenamefont {Sturm}(1999)}]{Sturm1999}%
  \BibitemOpen
  \bibfield  {author} {\bibinfo {author} {\bibfnamefont {J.~F.}\ \bibnamefont
  {Sturm}},\ }\href {\doibase 10.1080/10556789908805766} {\bibfield  {journal}
  {\bibinfo  {journal} {Optimization Methods and Software}\ }\textbf {\bibinfo
  {volume} {11}},\ \bibinfo {pages} {625} (\bibinfo {year} {1999})}\BibitemShut
  {NoStop}%
\bibitem [{\citenamefont {Takagi}\ and\ \citenamefont
  {Regula}(2019)}]{Takagi2019a}%
  \BibitemOpen
  \bibfield  {author} {\bibinfo {author} {\bibfnamefont {R.}~\bibnamefont
  {Takagi}}\ and\ \bibinfo {author} {\bibfnamefont {B.}~\bibnamefont
  {Regula}},\ }\href {\doibase 10.1103/PhysRevX.9.031053} {\bibfield  {journal}
  {\bibinfo  {journal} {Physical Review X}\ }\textbf {\bibinfo {volume} {9}},\
  \bibinfo {pages} {031053} (\bibinfo {year} {2019})}\BibitemShut {NoStop}%
\bibitem [{\citenamefont {Terhal}\ \emph {et~al.}(2003)\citenamefont {Terhal},
  \citenamefont {Doherty},\ and\ \citenamefont {Schwab}}]{Terhal2003}%
  \BibitemOpen
  \bibfield  {author} {\bibinfo {author} {\bibfnamefont {B.~M.}\ \bibnamefont
  {Terhal}}, \bibinfo {author} {\bibfnamefont {A.~C.}\ \bibnamefont {Doherty}},
  \ and\ \bibinfo {author} {\bibfnamefont {D.}~\bibnamefont {Schwab}},\ }\href
  {\doibase 10.1103/PhysRevLett.90.157903} {\bibfield  {journal} {\bibinfo
  {journal} {Physical Review Letters}\ }\textbf {\bibinfo {volume} {90}},\
  \bibinfo {pages} {157903} (\bibinfo {year} {2003})}\BibitemShut {NoStop}%
\bibitem [{\citenamefont {Vidal}\ and\ \citenamefont
  {Werner}(2002)}]{Vidal2002}%
  \BibitemOpen
  \bibfield  {author} {\bibinfo {author} {\bibfnamefont {G.}~\bibnamefont
  {Vidal}}\ and\ \bibinfo {author} {\bibfnamefont {R.~F.}\ \bibnamefont
  {Werner}},\ }\href {\doibase 10.1103/PhysRevA.65.032314} {\bibfield
  {journal} {\bibinfo  {journal} {Physical Review A}\ }\textbf {\bibinfo
  {volume} {65}},\ \bibinfo {pages} {032314} (\bibinfo {year}
  {2002})}\BibitemShut {NoStop}%
\end{thebibliography}%

\end{document}